\documentclass[12pt,a4paper]{article}

\usepackage[utf8]{inputenc}
\usepackage[margin=1in]{geometry}
\usepackage{lmodern}
\usepackage{tabularx}
\newcolumntype{L}{>{\raggedright\arraybackslash}X}
\usepackage[dvipsnames, table]{xcolor}
\usepackage{listings}
\usepackage{tcolorbox}
\usepackage{multicol}
\usepackage{booktabs}
\usepackage{multirow}
\usepackage[flushleft]{threeparttable}
\usepackage{array}          
\usepackage{enumitem}
\usepackage{amsthm}
\usepackage{amsmath}
\usepackage{amssymb}
\usepackage{mathtools}
\usepackage{mathrsfs}
\usepackage{thmtools, thm-restate}
\usepackage{bm}
\usepackage{bbm}
\usepackage{dsfont}
\usepackage{siunitx}
\usepackage{setspace}
\usepackage{hyperref}
\usepackage{cleveref}
\usepackage{svg}
\usepackage{graphicx}
\usepackage{wrapfig}
\usepackage{graph-separation}
\usepackage{tikz,tikz-cd}
\usetikzlibrary{calc,decorations.pathmorphing,shapes, arrows, positioning}
\usepackage{caption}
\usepackage{subcaption}
\usepackage{authblk}
\usepackage{algorithm}
\usepackage[noend]{algpseudocode}
\usepackage[style=authoryear,uniquename=false,maxcitenames=2,uniquelist=false,doi=false,url=false,isbn=false]{biblatex}
\AtEveryBibitem{%
  \clearfield{note}%
  \clearfield{language}%
}
\renewcommand{\arraystretch}{0.5}  

\hypersetup{
  colorlinks=true,
  linkcolor=blue,
  filecolor=magenta,
  urlcolor=cyan,
  pdftitle={Overleaf Example},
  pdfpagemode=FullScreen,
  citecolor=blue
}

\newtheorem{assumption}{\normalfont\scshape Assumption}

\newtheorem{manualtheoreminner}{\normalfont\scshape Assumption}
\newenvironment{assumption*}[1]{%
  \renewcommand\themanualtheoreminner{#1}%
  \manualtheoreminner
}{\endmanualtheoreminner}

\Crefname{assumption}{Assumption}{Assumptions}
\Crefname{assumption*}{Assumption}{Assumptions}

\newtheorem{theorem}{\normalfont\scshape Theorem}
\newtheorem{corollary}[theorem]{\normalfont\scshape Corollary}

\newtheorem{lemma}[theorem]{\normalfont\scshape  Lemma}
\newtheorem{proposition}[theorem]{\normalfont\scshape  Proposition}

\theoremstyle{definition}

\usepackage{apptools}
\AtAppendix{\counterwithin{theorem}{section}}
\AtAppendix{\counterwithin{figure}{section}}
\AtAppendix{\counterwithin{table}{section}}
\AtAppendix{\counterwithin{equation}{section}}

\definecolor{graphBlue}{RGB}{0, 0, 255}

\definecolor{graphGreen}{RGB}{0,100,0} 

\definecolor{graphRed}{RGB}{220, 20, 60}

\definecolor{graphOrange}{RGB}{255, 165, 0}

\definecolor{cornflowerblue}{rgb}{0.39, 0.58, 0.93}
\definecolor{olive}{rgb}{0.5, 0.5, 0.0}

\newcommand{\E}[0]{\mathbb{E}}

\newcommand{\PP}[0]{\mathbb{P}}

\newcommand{\pa}[0]{\text{pa}}

\newcommand{\ch}[0]{\text{ch}}
\newcommand{\an}[0]{\text{an}}
\newcommand{\de}[0]{\text{de}}

\newcommand{\indep}{\perp \!\!\! \perp}

\newcommand{\ingraph}[1]{~\textnormal{\bf in}\,#1}

\newcommand{\textnot}{\textnormal{\bf not}~}
\newcommand{\underdist}[1]{~\textnormal{\bf under}\,#1}

\DeclareMathOperator{\gG}{\mathscr{G}}

\DeclareMathOperator{\sD}{\mathcal{D}}
\DeclareMathOperator{\sB}{\mathcal{B}}

\DeclareMathOperator{\sJ}{\mathcal{J}}

\DeclareMathOperator{\sK}{\mathcal{K}}

\DeclareMathOperator{\sL}{\mathcal{L}}

\allowdisplaybreaks

\providecommand{\keywords}[1]
{
  \textbf{\textit{Keywords---}} #1
}

\title{A Graphical Approach to State Variable Selection in Off-policy Learning}

\author[1,2]{Joakim Blach
  Andersen\footnote{jb2413@cam.ac.uk}}
\author[1]{Qingyuan Zhao\footnote{qyzhao@statslab.cam.ac.uk}}
\affil[1]{Statistical Laboratory, University of Cambridge, UK}
\affil[2]{A.P.\ Moller Maersk}

\date{}

\begin{document}

\maketitle

\begin{abstract}

  Sequential decision problems are widely studied across many areas
  of science. A key challenge when learning policies from historical
  data---a practice commonly referred to as off-policy learning---is
  how to ``identify'' the impact of a policy of interest when the
  observed data are not randomized.
  Off-policy learning has mainly been studied in two settings: dynamic
  treatment regimes (DTRs), where the focus is on controlling
  confounding in medical problems with short decision horizons, and
  offline reinforcement learning (RL), where the focus is on dimension
  reduction in closed systems 
  such as games. The gap between these two well studied settings has
  limited the wider application of off-policy learning to many
  real-world problems.
  Using the theory for causal inference based on acyclic directed
  mixed graph (ADMGs), we provide a set of graphical identification
  criteria in general decision processes that
  encompass both DTRs and MDPs. We discuss how our results relate to
  the often implicit causal assumptions made in the DTR and RL literatures
  and further clarify several common misconceptions. Finally, we
  present a realistic simulation study for the dynamic pricing problem
  encountered in container logistics, and demonstrate how violations
  of our graphical criteria can lead to suboptimal policies.


\end{abstract}


\section{Introduction}\label{sec:introduction}


Sequential decision-making problems are routinely encountered in
many areas of science, engineering, and business. Among them, two
problems have been studied most extensively: Markov decision processes
(MDPs) in reinforcement learning (RL)---a popular setting in engineering
and computer science
\parencite{puterman_markov_2014,sutton_reinforcement_2018}---and dynamic
treatment regimes (DTRs) in biostatistics and health research
\parencite{murphy_optimal_2003,chakraborty_statistical_2013}. It
is widely acknowledged that
these two problems are closely related. For example, a recent
book-long treatment of DTRs by \textcite[page 187 and 574]{tsiatis_dynamic_2020} use
terminology from RL to describe their framework
and methodology. As another example, a recent review of off-policy evaluation in RL by
\textcite{uehara_review_2022} discusses how a DTR may be viewed as a finite-horizon MDP when several assumptions are relaxed.

Many real-world applications bear resemblence to both MDPs and
DTRs. As a motivating example, we consider the dynamic pricing problem
encountered in container logistics, where the decision-maker needs to
set and update prices for container shipments between two
destinations. On one hand, historical prices are set by shipping
professionals, making the dynamic pricing problem similar to learning
DTRs from existing health data where treatment decisions are made by
physicians. On the other hand, the pricing problem is repetitive and
has an inherently long horizon, naturally lending itself to an
infinite-horizon MDP. Dynamic pricing in container logistics also
bears similarities to many other sequential decision problems in the
real world---most notably pricing problems encountered in the airline
industry.


Surprisingly, there have been few attempts to unify the MDP and DTR literatures,
with a notable exception being
\textcite{ertefaie_constructing_2018}. The fundamental
``identification problem''---whether it is possible to
learn policies or treatment regimes from empirical data---is approached
very differently in these two fields. The DTR literature is pioneered by
\textcite{robins_new_1986} and emphasizes when a treatment regime can
possibly be evaluated using experimental or observational data. As a
result, most papers on this topic start with a version of ``no
confounding'' or ``sequential
ignorability'' assumption before discussing any further theory and
methodology. Contrary to this, there is virtually no discussion on the
feasibility of off-policy learning in computer science and
engineering. Instead, in most methodological work for MDPs, it is
assumed that one is given a collection of variables---referred to
as the \emph{state}---that fully summarizes the system dynamics
and allows the decision-maker to discard the rest of the history.
In practice, however, the state must be selected from
a set of observed variables. And while
state variable selection may be straightforward in closed-system games
such as chess or Go, selecting the appropriate state in real-world
problems such as dynamic pricing for container shipments is critical
and far from obvious. 
This lack of attention to identifiability and absence of practical
guidance on state variable selection pose significant challenges to applying
RL to real-world problems.

\subsection{The identifiability problem and main result}
The main purpose of this article is to develop a unifying framework of
sequential decision problems that encompasses DTRs and MDPs, and to
discuss state variable selection in this general setting. We study the
central ``identification'' question in such problems: when is it possible
to estimate the value of a policy of interest from data generated by a
different policy? Our main contribution is a set of graphical
identification conditions that extend the ``backdoor criterion'' of
\textcite{pearl1993} to the dynamic setting and the memorylessness
assumption commonly used in the MDP literature.


To state our graphical criteria and the main identifiability
result, let us briefly introduce some notation. Consider a
(discrete-time) \emph{decision process} defined as a collection of
random variables
\[
  V = (X_1, A_1, X_2, A_2, \dots, X_T, A_T, X_{T+1}),~T \geq 1,
\]
where $X_t$ denote the set of variables that contain information
observed after decision $A_{t-1}$ and before decision $A_t$ for $t \in
[T] := \{1, 2, \dots,T\}$, and
$X_{T+1}$ is the information observed after $A_T$. We will assume $V$
has a probability density function with respect to some dominating
measure (e.g.\ the Lebesgue measure if the $V$ is continuous or the
counting measure if $V$ is discrete) and denote the value of the density at a point $v$ as $\PP(V =
v)$ or simply $\PP(v)$ if no confusion arises. We use overline to
indicate the history of that variable up to time $t$; for
example, $\overline{A}_t = (A_1, A_2, \dots, A_t)$ and $\overline{a}_t =
(a_1,a_2,\dots, a_t)$. Thus, $\PP(\overline{a}_t)$ means the density of
$\overline{A}_t$ at $\overline{a}_t$.

An \emph{adaptive
  policy} $g = (g_1,\dots,
g_T)$ intervenes on $A_t$ according to a probability density function
$g_t(a_t \mid s_t)$ that can only depend on a set of \emph{state
  variables} $S_t$ for $t \in [T]$. It is required that $S_t$ only
contains information before $A_t$, that is, $S_t \subseteq X_1 \cup
\{A_1\} \cup \dots \cup X_{t-1} \cup \{A_{t-1}\} \cup X_t$.
Given a policy, we are interesting in identifiying its reward as
measured by some variables $R_t \subseteq X_{t} \cup A_{t-1}$, $t \in [T+1]$.
For example, we may wish to estimate the average \emph{value} of some
``utility function'' of the rewards $R = (R_1,\dots,R_{T+1})$ under the
distribution induced by $g$,
\[
  \rho(\PP(g)) = \E[ u(R(g)) ],
\]
where $R(g)$ is the ``potential outcomes'' of $R$ under policy $g$, $u$
is a real-valued utility function, and $\PP(g)$ is the probability
distribution of the potential outcomes of $V$ under policy $g$. In the
RL literature, a common utility function is $u(R(g)) =
\sum_{t=0}^T \gamma^t R_t(g)$ where $0 < \gamma < 1$ is some dicount
factor. For the purpose of causal identification, we assume
the probability distribution $\PP$ of $V$ under some null
policy is known (with our notation we may regard $\PP$ as
$\PP(\text{null}, \dots,\text{null})$). In practice, we may need to use empirical data to
estimate $\PP$.

The central identifiabiliy question we will try to answer is:
\begin{center} \bf
    When can we equate $\PP(g)$ or $\rho(\PP(g))$ with quantities that
    only depend on $\PP$? 
\end{center}
To answer this, it is useful to use a partition of the variables
because the state, decision, and reward variables can
overlap. Specifically, let the ``innovations'' after decision
$A_{t-1}$ and before $A_t$ that are relevant to the decision problem
be defined as
\begin{equation}
  \label{eq:innovations}
  N_t = (R_t \cup S_t) \cap X_t,~t \in [T+1].
\end{equation}
Note that $N_t$ might be empty if $R_t = S_t \cap X_t =
\emptyset$, where $R_t = \emptyset$ means we are not interested in the
intermediate reward at time $t$. We use the convention $S_{T+1}
=\emptyset$.

To answer the identifiability question we need to describe causal
relationships between the variables in $V$. To this end, we will
use an \emph{acyclic directed mixed graph} (ADMG) with vertex
set $V$ that have two types of edges: directed ($\rdedge$) and
bidirected ($\bdedge$). Such graphs were first used by
\textcite{wright_1934} and play a central role in the
statistical theory for causality; see, for example,
\textcite{pearl2009} 
and \textcite{richardson_nested_2023} (although Pearl use a different
terminology). 
A central concept in ADMG models is
\emph{m-separation} \parencite{richardson_markov_2003} that extends
d-separation for \emph{directed acyclic graphs} (DAGs)
\parencite{pearl2009}. Two vertices in the graph are said to be
m-separated by $L \subseteq V$ if every path between them is \emph{ancestrally
  blocked} by $L$,\footnote{Most authors simply say such a path is blocked. We
used ``blocked'' for a slightly different concept for walks in the
graph; see \Cref{subsec:m-separation}.} which means
\begin{enumerate}
    \item[i)] the path contains a collider (any vertex $V_j$ that
      looks like $\rdedge V_j \ldedge$, $\bdedge V_j \ldedge$, $\rdedge V_j \bdedge$,
      $\bdedge V_j \bdedge$) that is not an ancestor of
      $L$; \text{or}
    \item[ii)] the path contains a non-collider that is in $L$.
    \end{enumerate}
\Cref{sec:prelim} will review the relevant concepts and results in the
causal ADMG theory. Readers who are not familiar with this theory can simply interpret m-separation as a graphical notion of variable
independence. 

We are now
ready to give our new identifiability conditions.
\begin{assumption}[Nested states] \label{assump:nested-states}
  We assume $S_t \subseteq S_{t-1} \cup A_{t-1} \cup X_t$  for all $t \in [T]$.
\end{assumption}

\begin{assumption}[Memorylessness] \label{assump:memorylessness}
  In the causal ADMG, $(\overline{S}_{t-1} \cup \overline{A}_{t-1}) \setminus
  S_t$ and $N_{t+1}$ are m-separated by $S_t$ for all $t \in [T]$.
\end{assumption}


\begin{assumption}[Dynamic back-door] \label{assump:no-confounding-a}
  In the causal ADMG, every path from $A_k$ to $N_{t+1}$ with an
  arrowhead into $A_k$ is ancestrally blocked by $S_k$ for all $k \leq
  t \leq T$.
\end{assumption}

Assumption \ref{assump:nested-states} assumes that a variable that is
left out in a previous state cannot be part of a new state.
Assumption \ref{assump:memorylessness} states that the novel variables
at time $t+1$ must be \emph{m-separated} from previous state and
action variables given the current state $S_t$. It is helpful to think
of Assumption \ref{assump:memorylessness} as a graphical counterpart
of the memorylessness assumption that is common in the MDP
literature. 
Assumption \ref{assump:no-confounding-a} extends the back-door
criterion for unconfoundedness in \textcite{pearl1993} to the dynamic
setting. 

Our main result is that the above assumptions are basically sufficient
to identify the joint distribution of the rewards, states, and actions
under the adaptive policy $g$. As in the static case ($T = 1$), we
will need a dynamic consistency property of the potential outcomes and a
positivity assumption for $\PP$. These two assumptions will be
introduced as \Cref{assump:dynamic-consistency,assump:positivity}
later in the article.
\begin{theorem}\label{thm:main_result}
    Let dynamic consistency and positivity be given, and suppose $R_t
    \subseteq S_t$ for $t \in [T]$. Under Assumptions \ref{assump:nested-states},
    \ref{assump:memorylessness}, and \ref{assump:no-confounding-a},
    we have
    \begin{align*}
      \mathbb{P}(\overline{N}_{T+1}(g) =
      \overline{n}_{T+1}, \overline{A}_{T}(g) =
      \overline{a}_{T}) = \mathbb{P}(r_{T+1} \mid a_T, s_T)
      \prod_{t=1}^T g(a_{t} \mid s_{t}) \PP(n_{t} \mid a_{t-1},
      s_{t-1}).
    \end{align*}
\end{theorem}
Note that $\overline{N}_{T+1} = (\overline{S}_T \cup
\overline{R}_{T+1}) \setminus \overline{A}_T$ by \eqref{eq:innovations} and
the nested states assumption, so the left hand side of this equation
is the probability density of all state, decision, and reward
variables. As another remark, the condition $R_t \in S_t$, $t \in [T]$
in \Cref{thm:main_result} can be dropped if one is just interested in
identifying the marginal distribution of $R_{T+1}(g)$ (so
$R_1,\dots,R_T$ are empty).

\subsection{Motivating example: dynamic pricing for container
  logistics}

\label{sec:motivating-example}

Let us illustrate the graphical criteria in \Cref{thm:main_result}
using an example motivated by the dynamic pricing problem for
container logistics. In this simplified example represented by the
graph in \Cref{fig:toy_pricing_template}, a company sets prices
for container slots on weekly vessel departures. At each time $t$, the
company sets a price $A_t$ for a standard 40-foot container and
receives bookings $B_{t+1}$ for the vessel departing at time
$t+1$. The total revenue for this vessel is thus $R_{t+1} = A_t
B_{t+1}$. To decide the price, suppose shipping professionals in
the company consider the previous price (represented by the directed
edge $A_{t-1} \rdedge A_t$) and bookings ($B_t
\rdedge A_t$). The reader can easily verify
that the choice $S_t =
\{A_{t-1}\}$ satisfies Assumptions \ref{assump:nested-states},
\ref{assump:memorylessness} and \ref{assump:no-confounding-a} in
\Cref{fig:toy_pricing_template}.

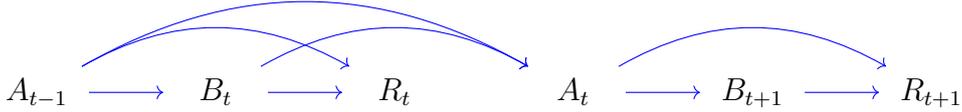
\begin{figure}[t]
  \centering
  \begin{tikzcd}[nodes={text width=1cm, align=center}]
    A_{t-1} \arrow[r, graphBlue] \arrow[rr, graphBlue, bend left] \arrow[rrr, graphBlue, bend left] & B_{t} \arrow[r, graphBlue] \arrow[rr, graphBlue, bend left] & R_t & A_t \arrow[r, graphBlue] \arrow[rr, graphBlue, bend left] & B_{t+1} \arrow[r, graphBlue] & R_{t+1}
  \end{tikzcd}
  \caption{Simple dynamic pricing.}
  \label{fig:toy_pricing_template}
\end{figure}

\Cref{fig:simple_dynamic_pricing} presents
three plausible deviations from the simple setting in
\Cref{fig:toy_pricing_template}. 
In \Cref{subfig:competitor_price},
shipping professionals rely on word-of-mouth intel about competitor
prices when setting their own prices. Because competitor prices also
impact the number of bookings the company receives, this introduces
confounding between the price and the number of bookings (as
represented by the bidirected edge $A_{t} \bdedge B_{t+1}$). In
\Cref{subfig:lagged_price}, customers
consider previous prices when making their booking decision (as
represented by the additional edge $A_{t-1} \rdedge B_{t+1})$ but the
shipping professionals follow a myopic pricing policy (as represented
by the lack of the edges $B_{t} \rdedge A_{t}$ and $A_{t-1} \rdedge A_{t}$). Finally,
\Cref{subfig:route_capacity} presents a scenario with two latent
confounders: market trend and port congestion. Suppose the shipping
company adjusts the total vessel capacity 
$C_t$ based on the market trend, which 
affect the received bookings at the same time (as represented by $C_{t} \bdedge
B_{t+1}$). 
Furthermore, port congestion may limit the weekly container supply but
also increase the price through increased operational costs ($A_{t}
\bdedge C_{t}$).
\begin{figure}[t]
  \centering
  \begin{subfigure}[b]{0.8\textwidth}
    \centering
    \begin{tikzcd}[nodes={text width=1cm, align=center}]
      A_{t-1} \arrow[r, graphBlue] \arrow[r, graphRed, leftrightarrow, bend right] \arrow[rr, graphBlue, bend left] \arrow[rrr, graphBlue, bend left] & B_{t} \arrow[r, graphBlue] \arrow[rr, graphBlue, bend left] & R_t & A_t \arrow[r, graphBlue] \arrow[r, graphRed, leftrightarrow, bend right] \arrow[rr, graphBlue, bend left] & B_{t+1} \arrow[r, graphBlue]  & R_{t+1}
    \end{tikzcd}
    \caption{Competitor prices. Not identified.}
    \label{subfig:competitor_price}
  \end{subfigure}

  \begin{subfigure}[b]{0.8\textwidth}
    \centering
    \begin{tikzcd}[nodes={text width=1cm, align=center}]
      A_{t-1} \arrow[r, graphBlue] \arrow[rr, graphBlue, bend left]  \arrow[rrrr, graphRed, bend left] & B_{t} \arrow[r, graphBlue] & R_t  & A_{t} \arrow[r, graphBlue] \arrow[rr, graphBlue, bend left] & B_{t+1} \arrow[r, graphBlue] & R_{t+1}
    \end{tikzcd}
    \caption{Lagged price effects. The state choice $\hat{S}_t =
      \emptyset$ violates Assumption \ref{assump:memorylessness}. The
      state choice $S_t^{*} = \{A_{t-1}\}$ satisfy all idenfiability
      assumptions.}
    \label{subfig:lagged_price}
  \end{subfigure}


  \begin{subfigure}[b]{0.8\textwidth}
    \centering
    \begin{tikzcd}[nodes={text width=1cm, align=center}]
      C_{t-1} \arrow[d, graphRed, leftrightarrow] \arrow[rd, graphRed, leftrightarrow, bend left]& & & C_t \arrow[d, graphRed, leftrightarrow] \arrow[rd, graphRed, leftrightarrow, bend left] & & \\
      A_{t-1} \arrow[r, graphBlue] \arrow[rr, graphBlue, bend left] \arrow[rrr, graphBlue, bend left] & B_{t} \arrow[r, graphBlue] \arrow[rr, graphBlue, bend left] & R_t & A_t \arrow[r, graphBlue] \arrow[rr, graphBlue, bend left] & B_{t+1} \arrow[r, graphBlue] & R_{t+1}
    \end{tikzcd}
    \caption{Route capacity. The state choice $\hat{S}_t = \{A_{t-1},
      C_t\}$ violates Assumption \ref{assump:no-confounding-a}. The
      state choice $S_t^{*} = \{A_{t-1}\}$ satisfy all idenfiability
      assumptions.}
    \label{subfig:route_capacity}
  \end{subfigure}


  \caption{Motivating dynamic pricing examples. New edges compared to
    \Cref{fig:toy_pricing_template} are shown in red.}
  \label{fig:simple_dynamic_pricing}
\end{figure}
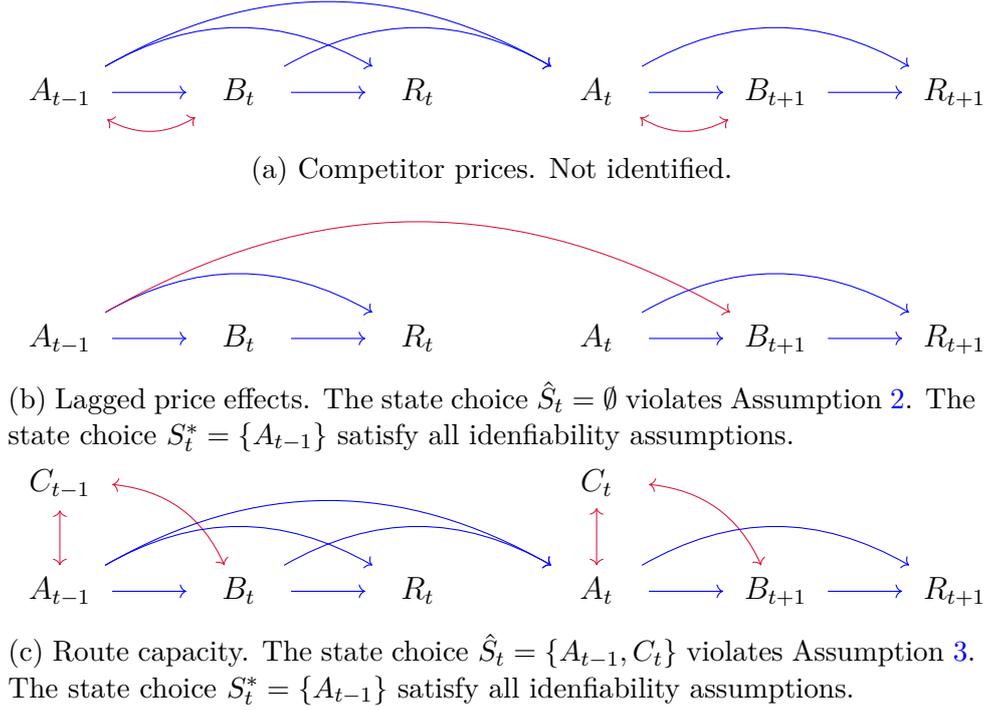

What choice of the state variables will satisfy Assumptions
\ref{assump:nested-states}, \ref{assump:memorylessness} and
\ref{assump:no-confounding-a} in the three scenarios in
\Cref{fig:simple_dynamic_pricing}? It is obvious that no such choice
exists for \Cref{subfig:competitor_price} because the immediate
confounding between $A_{t-1}$ and $B_t$ cannot be controlled for. In
other words, \Cref{assump:no-confounding-a} can never be true. In
\Cref{subfig:lagged_price}, the state choice $S_t = \emptyset$ satisfies
nestedness and unconfoundedness, but not memorylessness due to the edge
$A_{t-1} \rdedge B_{t+1}$. 
In
\Cref{subfig:route_capacity}, the state choice $S_t = \{A_{t-1},
C_t\}$ satisfies nestedness and memorylessness, but not
unconfoundedness due to the collider path $A_{t-1} \bdedge C_{t-1}
\rdedge B_t$. The state choice $S_t = \{A_{t-1}\}$ satisfies all
three identifiability assumptions for both \Cref{subfig:lagged_price} and \ref{subfig:route_capacity}.

Using a state that does not satisfy the identifiability assumptions
may lead to inferior policies. To demonstrate this, we generated data
from a null policy according to each scenario in
\Cref{fig:simple_dynamic_pricing} and used policy iteration to learn a
policy with different state choices. Policy iteration is a popular
RL algorithm that iterates between policy evaluation and
policy improvement; it provably converges to the optimal policy if
the MDP assumptions are satisfied
\parencite{howard:dp,sutton_reinforcement_2018}. However, it is often
overlooked that policy iteration may converge to a sub-optimal
policy if the state is not chosen correctly. In fact,
\Cref{fig:toy_pricing_results} shows that the policy learned by the
policy iteration algorithm using a state that does not satisfy
Assumptions \ref{assump:nested-states}-\ref{assump:no-confounding-a}
can be even worse than the null policy.

\begin{figure}[t]
  \centering
  \includegraphics[width=\textwidth]{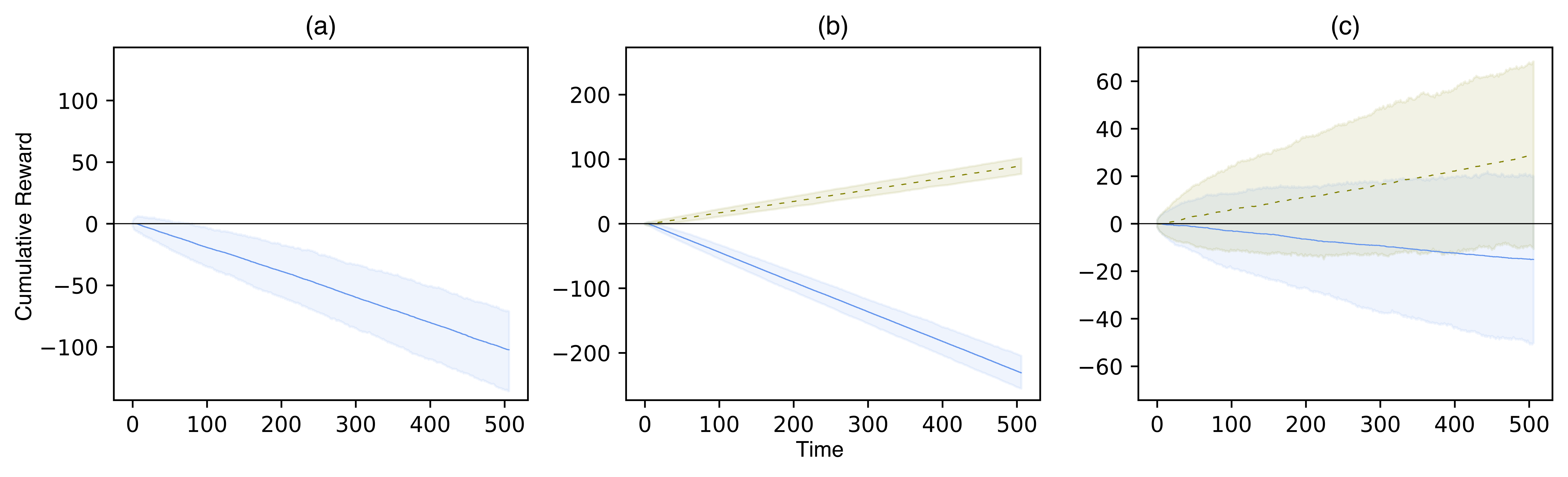}
  \caption{Mean cumulative rewards and 95\% confidence intervals
    relative to the null policy performance (equal to
    zero in the plots) using 1000 simulated episodes. In each panel,
    data are
    generated from a null policy according to the corresponding panel in
    \Cref{fig:simple_dynamic_pricing}.
    The solid blue curves correspond to the policy learned from
    policy iteration using a state $S_t$ that fails at least one
    identifiability assumptions ($S_t = \{A_{t-1}\}$ in panel (a),
    $S_t = \emptyset$ in panel (b), and $S_t = \{A_{t-1}, C_t\}$ in
    panel (c)). The dashed olive curves correspond to the policy
    learned using a state $S_t$ that satisfies
    Assumptions
    \ref{assump:nested-states}-\ref{assump:no-confounding-a} ($S_t =
    \{A_{t-1}\}$ in panels (b) and (c)).
  }
  \label{fig:toy_pricing_results}
\end{figure}

\subsection{Organization of the paper}
In \Cref{sec:prelim} we introduce some background on causal graphical
models including m-separation, the construction of d-SWIGs
\parencite{richardson_single_2013} and the dynamic
  consistency assumption. In \Cref{sec:identication} we provide a more
  technical walkthrough of \Cref{thm:main_result}, our main
  identification result. In
\Cref{sec:dtr} we discuss how our identifiability assumptions relate
to the common sequential ignorability assumption in the DTR
literature. In \Cref{sec:mdp}, we discuss the implicit causal
assumptions made in the MDP literature and how to interpret them in
the light of our results. We also discuss the abuse of
causal diagram in that literature and why explain partially observed
MDPs (POMDPs) are not identified in general. In
\Cref{sec:simulation_study} we present a more realistic simulation
study of the dynamic pricing problem from container logistics, and
examine how plausible violations of our assumptions can lead to
suboptimal policies. Finally in \Cref{sec:discussion} we conclude the
paper with some more dicussion. Technical proofs and details of the
simulation study can be found in the Online Supplement.




\section{Preliminaries of causal graphical models}\label{sec:prelim}



\subsection{Basic graphical concepts}

A directed mixed graph $\gG = (V, \sD, \sB)$ consists of a vertex set
$V =
\{V_1,\dots,V_p\}$, a \emph{directed edge} set $\sD \subseteq V \times
V$, and a \emph{bidirected edge} set $\sB \subseteq V \times V$ that
is required to be symmetric: $(V_j,V_k) \in \sB \Longleftrightarrow
(V_k,V_j) \in \sB,~\text{for all}~V_j,V_k \in V$.
It is helpful to think about the edges as relations between the
vertices and write
\[
    V_j \rdedge V_k \ingraph{\gG} \Longleftrightarrow (V_j,V_k) \in \sD\quad\text{and}\quad V_j \bdedge V_k \ingraph{\gG} \Longleftrightarrow (V_j,V_k) \in \sB.
\]
\paragraph{Walks and paths.}
A \emph{walk} is a sequence of adjacent edges of any type or orientation. If the two-endpoints appear only once we say the walk is \emph{simple}, and if all vertices appear at most once in the walk we call it a \emph{path}. \vspace{-2ex}
\paragraph{Colliders.} A non-endpoint $V$ in a
walk $w$ is said to be a \emph{collider} if the two edges before and
after $V$ have an arrowhead into $V$. 
When describing a walk or part of a walk, we use a "half arrowhead"
to indicate that the endpoint of an edge can be either a arrowhead or
a tail. For example, $V$ is a collider in $w$ if $w$ contains
$\halfstraigfull V \fullstraighalf$. It is obvious that the same
vertex can be a collider in one walk (or one place in the walk) and
non-collider in another walk (or another place in the same
walk). 
\paragraph{Arcs, directed walks, and confounding
    arcs.} An \emph{arc} is a walk without colliders. Following the
  notation in \textcite{zhao_matrix_2024}, we denote an arc by a
squiggly line ($\nosquigno$). 
We further distinguish
arcs by embellishing their endpoints with no, half or full
arrowhead. For example, a \emph{directed walk} from $V_j$ to $V_K$ is
a sequence of adjacent directed edges
$V_j \rdedge \dots \rdedge V_k$. When such walks exist, we write $V_j
\rdpath V_k$. A \emph{confounding arc}
between $V_j$ and $V_k$, which looks like $V_j \fullsquigfull V_k$, is
a simple walk
with no colliders and two endpoint arrowheads. As before, $V_j
\halfsquigfull V_k$ means that the walk is either a directed walk or a
confounding arc. If a directed walk has the same beginning and end, it
is called a directed cycle. If a directed mixed graph has no directed
cycles, we say it is \emph{acyclic} or an ADMG.
\paragraph{Familial terminology.} If $V_j \rdedge V_k$ then
$V_j$ is a \emph{parent} of $V_k$, and $V_k$ is a \emph{child} of $V_j$. If $V_j
\rdpath V_k$, then $V_j$ is an \emph{ancestor} of $V_k$ and $V_k$ is a
\emph{descendant} of $V_j$. We will use the convention that every vertex is
an ancestor and descendant of itself. The sets of parents, children,
ancestors and denscendants of $V_j$ in $\gG$ are denoted $\pa(V_j)$,
$\ch(V_j)$, $\an(V_j)$ and $\de(V_j)$, respectively. 

\subsection{m-separation and the Markov property}\label{subsec:m-separation}
A central concept in graphical statistical models is blocking. We say
a walk $w$ from $V_j \in V$ to $V_k \in V$ is \emph{blocked} by $L \subseteq V$ if
\begin{enumerate}
    \item $w$ contains a collider $V_l$ (so part of $w$ looks like
  $\halfstraigfull V_l \fullstraighalf$) and $V_l \not \in L$; \text{or}
    \item $w$ contains a non-collider $V_l$ (which must also be a
      non-endpoint) such that $V_l \in L$.
\end{enumerate}
Note that with this definition a walk cannot be blocked at its
endpoints, that is, $w$ is not necessarily blocked by $L$ if $L$
contains an endpoint of $w$. If $V_j$, $V_k$, and $L$ are disjoint, we
say $V_j$ is \emph{m-connected} to $V_k$ given $L$ and write $V_j
\mconn V_k \mid L \ingraph{\gG}$, if there exists an unblocked walk from $V_j$ to $V_k$ given $L$; otherwise we say $V_j$ and $V_k$ are \emph{m-separated} given $L$ in $\gG$ and write $\textnot V_j
\mconn V_k \mid L \ingraph{\gG}$. M-separation is introduced by
\textcite{richardson_markov_2003} and extends the
d-separation criterion for conditional independence in DAGs
\parencite{pearl1988book} to ADMGs. 
The definition of m-separation here using walks and blocking is
equivalent to the definition in \Cref{sec:introduction} using paths and
ancestral blocking. See
\textcite{shachterBayesballRationalPastime1998,guo_confounder_2023}.

Our notation gives a visual description of the type of walk: (1) the
half-arrowheads indicate that both endpoints
are unrestricted in terms of arrowhead or tail, and (2) the wildcard
character $\ast$ means the walk can have zero, one or several
colliders. Thus, $V_j \mconn V_k$ basically refers to all
walks from $V_j$ to $V_k$. This definition of m-connection/separation
naturally extends to sets of vertices: for disjoint $J, K, L \subset
V$, we write
\[
  J \mconn K \mid L \ingraph{\gG} \Longleftrightarrow
  V_j \mconn V_k \mid L \ingraph{\gG}~\text{for some}~V_j \in
  J, V_k \in K.
\]
It is often useful to consider other types of walks. For example, a
confounding walk can be expressed as $\confpath$
\parencite{guo_confounder_2023}. See the Online Supplement for further
discussion.

A probability distribution $\PP$ on $V$ is said to be \emph{global
  Markov} with respect to $\gG$ if every m-separation in the graph
implies the corresponding conditional independence, that is, if for
every disjoint $J,K,L \subset V$, we have
\[
  \textnot J \mconn K \mid L \ingraph{\gG} \Longrightarrow J \indep K
  \mid L \underdist{\PP}.
\]
An ADMG may impose other constraints on the probability distribution,
most notably the nested Markov property that is closely related to
causal identification \parencite{richardson_nested_2023}.

An experienced reader may find the above definition of blocking
slightly different from many other authors who define graph separation
using paths \parencite{lauritzen_graphical_1996, pearl2009}, which we
refer to as ancestral blocking. The advantage
of using this alternative notion of blocking is that it is entirely a
property of the walk and the set of vertices being conditioned
on; this is known as the Bayes ball algorithm in the literature
\parencite{shachterBayesballRationalPastime1998}. In contrast,
ancestral blocking depends on
the ambient graph and is less convenient in mathematical proofs. 

\subsection{Graphical causal models}
We will now introduce a formal causal model associated with an ADMG
$\gG$ with vertex set $V = \{V_1,\dots,V_p\}$. A nonparametric
structural equation model (NPSEM) with respect to $\gG$ collects all
distributions $P$ of $V$ such $V$ can be written as (following event
has probability 1 under $P$):
\begin{equation} \label{eq:npsem}
  V_j = f_j(V_{\pa(j)}, E_j), ~ j = 1,\dots,p,
\end{equation}
for some functions $f_1,\dots,f_p$ and unobserved noise variables
$E_1,\dots,E_p$ whose distribution satisfies the global Markov
property with respect to the bidirected subgraph of $\gG$. Because any
walk in the bidirected subgraph must be a sequence of bidirected
edges, this means
\[
  \textnot V_{\sJ} \samedist V_{\sK} \mid V_{\sL} \ingraph{\gG}
  \Longrightarrow E_{\sJ} \indep E_{\sK} \mid E_{\sL},~\text{for all
    disjoint}~\sJ, \sK, \sL \subset \{1,\dots,p\},
\]
where $V_{\sJ} \samedist V_{\sK} \mid V_{\sL}$ means that a vertex in
$V_{\sJ}$ and a vertex in $V_{\sK}$ can be connected by a walk that is
not blocked by $V_{\sL}$; when
$V_{\sL} = \emptyset$, this means that the two vertices are in the
same \emph{district} in the terminology of
\textcite{richardson_markov_2003} and \textcite{richardson_nested_2023}.

The potential outcomes under an (adaptive) policy can be defined by
modifying the equations in \eqref{eq:npsem}. To formalize this, let
$\prec$ be a topological order of $\gG$ in the sense that $V_j \rdedge
V_k \ingraph{\gG}$ implies $V_j \prec V_k$. Let $A = (A_1,\dots, A_T)
\subseteq V$ be the decision variables that can be changed by a
policy, and let $X = (X_1, \dots, X_{T+1})$ be the observed
information between the decisions. Together, they induce a natural
partition of $V$ through
\[
  X_1 \prec A_1 \prec X_2 \prec \dots \prec X_T \prec A_T \prec X_{T+1}.
\]
A \emph{policy} is defined as $g = (g_1, \dots, g_T)$ where $g_t$
is a function of some \emph{state} $S_t \prec A_t$ that precedes $A_t$
and some noise $E^*_t$ such that $(E_1^{*}, \dots, E_p^{*})$ is
independent of $(E_1, \dots, E_p)$. The potential outcome of $V_j$
under $g$ is defined recursively as
\[
    V_j(g) =
    \begin{cases}
      f_j(V_{\pa(j)}(g), E_j),& \text{if}~V_j \not \in A, \\
      g_t(S_t(g), E^*_t),& \text{if}~V_j = A_t~\text{for some $t$}. \\
    \end{cases}
\]
Thus, we simply replace the equations for $A_t$ in \eqref{eq:npsem} by $A_t = g_t(S_t(g), E_t^{*})$ and rename all variables to emphasize their dependence on $g$. With an abuse of notation, we will also use $g_t(a_t \mid s_t)$ to denote
the density of function of $A_t(g)$ at $a_t$ given $S_t(g) = s_t$.
It is also useful to define the ``natural counterfactuals'' of the decision variables as\footnote{One can similarly define the natural counterfactuals for all variables. However, $V_j(g) = V_j^-(g)$ if $V_j \not \in A$ (because no intervention is made on such $V_j$), so it is not useful to distinguish $V_j(g)$ and $V_j^-(g)$ when $V_j$ is not a decision variable.}
\[
  A_t^-(g) = f_j(V_{\pa(j)}(g), E_j), ~ \text{if} ~ V_j = A_t~\text{for some $t$}.
\]
Note that $A_j^-(g)$ represents a true counterfactual in sense that it
models what $A_t$ would have been \emph{had we not intervened on
  $A_t$} but at all other time-points according to $g$.

%



Following \textcite{richardson_single_2013}, one can represent
the effect of a policy $g$ as a graph transformation, in which every
decision vertex $A_t$ in $\gG$ is split into two halves:
\begin{enumerate}
    \item the natural counterfactual $A_j^-(g)$ that inherits all "incoming" edges $V_k \halfstraigfull A_t$;
    \item the potential outcome $A_j(g)$ that inherits all "outgoing" edges $A_t \nostraigfull V_k$.
\end{enumerate}
The natural counterfactual $A_t^-(g)$ has the same parent set as that
of $A_t$ in $\gG$ and has no children.
The potential outcome $A_t(g)$ has a parent set that is determined by the
policy $g_t$ and must be contained in $S_t(g)$, and has the same child
set as that of $A_t$ in $\gG$. Every other edge is kept the same and
every other variable $V_j$ is relabelled as $V_j(g)$. The resulting
graph is called a \emph{dynamic single-world intervention graph} or
d-SWIG by \textcite[Section 5]{richardson_single_2013} and
will be denoted as
$\gG(g)$.\footnote{\textcite[Section 5]{richardson_single_2013} only
  considered the case where $\gG$ is a DAG, but the extension to ADMGs
  here is fairly
  straightforward. \textcite{richardson_single_2013} denoted the
  natural counterfactual of $A_t$---our $A_t^-(g)$---as $A_t(g)$ and
  the potential outcome of $A_t$---our $A_t(g)$---as
  $A_t^+(g)$.} Note that in our setting, there
are no bidirected edges in $\gG(g)$ with one end being $A_t(g)$, which
reflects the assumption that $A_t(g)$ is randomized given $S_t(g)$
according to $g$. Figure \ref{fig:d-SWIG-examples} presents an example
with three different policies.
\begin{figure}[t]
    \begin{subfigure}[b]{\textwidth}
        \centering
        \begin{tikzcd}
            L_1 \arrow[rr, color=graphBlue, bend left] \arrow[r, color=graphBlue] & A_1 \arrow[r, color=graphBlue] \arrow[rr, color=graphRed, leftrightarrow, bend left] & L_2 \arrow[r, color=graphBlue] \arrow[rr, color=graphBlue, bend left] & A_2 \arrow[r, color=graphBlue] & R_3
        \end{tikzcd}
        \caption{$\gG$. No intervention.}
        \label{subfig:d_SWIG_a}
    \end{subfigure}

    \begin{subfigure}[b]{\textwidth}
        \centering
        \begin{tikzcd}
            & A^{-}_1(g_1) \arrow[rr, leftrightarrow, color=graphRed, bend left] &           &     A^{-}_2(g_1,g_2)      & \\
            L_1 \arrow[ru, color=graphBlue, bend left] \arrow[rr, color=graphBlue, bend left] \arrow[r, dashed, color=graphGreen] & \textcolor{graphGreen}{A_1(g_1)} \arrow[r, color=graphBlue] \arrow[rr, color=graphGreen, dashed, bend right]    & L_2(g_1) \arrow[r, dashed, color=graphGreen] \arrow[rr, color=graphBlue, bend left] \arrow[ru, color=graphBlue, bend left] & \textcolor{graphGreen}{A_2(g_1,g_2)} \arrow[r, color=graphBlue] & R_3(g_1,g_2)
        \end{tikzcd}
        \caption{$\gG(g_1, g_2)$ with $S_1 = \{L_1\}$, $S_2 = \{A_1, L_2\}$.}
        \label{subfig:d_SWIG_c}
      \end{subfigure}

    \begin{subfigure}[b]{\textwidth}
        \centering
        \begin{tikzcd}
            &  &           &     A^{-}_2(g_2)      & \\
            L_1 \arrow[r, color=graphBlue] \arrow[rr, color=graphBlue, bend left] \arrow[rrr, color=graphGreen, dashed, bend right] & A_1  \arrow[rru, leftrightarrow, color=graphRed, bend left] \arrow[rr, color=graphGreen, dashed, bend right] \arrow[r, color=graphBlue]     & L_2 \arrow[r, dashed, color=graphGreen] \arrow[rr, color=graphBlue, bend left] \arrow[ru, color=graphBlue, bend left] & \textcolor{graphGreen}{A_2(g_2)} \arrow[r, color=graphBlue] & R_3(g_2)
        \end{tikzcd}
        \caption{$\gG(\text{null}, g_2)$ with $S_2 = \{X_1, A_1, X_2\}$}
        \label{subfig:d_SWIG_b}
    \end{subfigure}

    \begin{subfigure}[b]{\textwidth}
        \centering
    \begin{tikzcd}
        &  A^{-}_1(g_1) \arrow[rrd, graphRed, leftrightarrow, bend left] &           &          & \\
        L_1 \arrow[ru, graphBlue] \arrow[r, dashed, graphGreen] \arrow[rr, graphBlue, bend left] & \textcolor{graphGreen}{A_1(g_1)}   \arrow[r, graphBlue]     & L_2(g_1) \arrow[r, graphBlue] \arrow[rr, graphBlue, bend left] & A_2(g_1) \arrow[r, graphBlue] & R_3(g_1)
    \end{tikzcd}
    \caption{The d-SWIG $\gG(g_1, \text{null})$ with $S_1 = \{L_1\}$.}
    \label{fig:d-SWIG-g1}
    \end{subfigure}

    \caption{d-SWIG examples; ``null'' means no intervention
      for the corresponding decision.}
    \label{fig:d-SWIG-examples}
\end{figure}
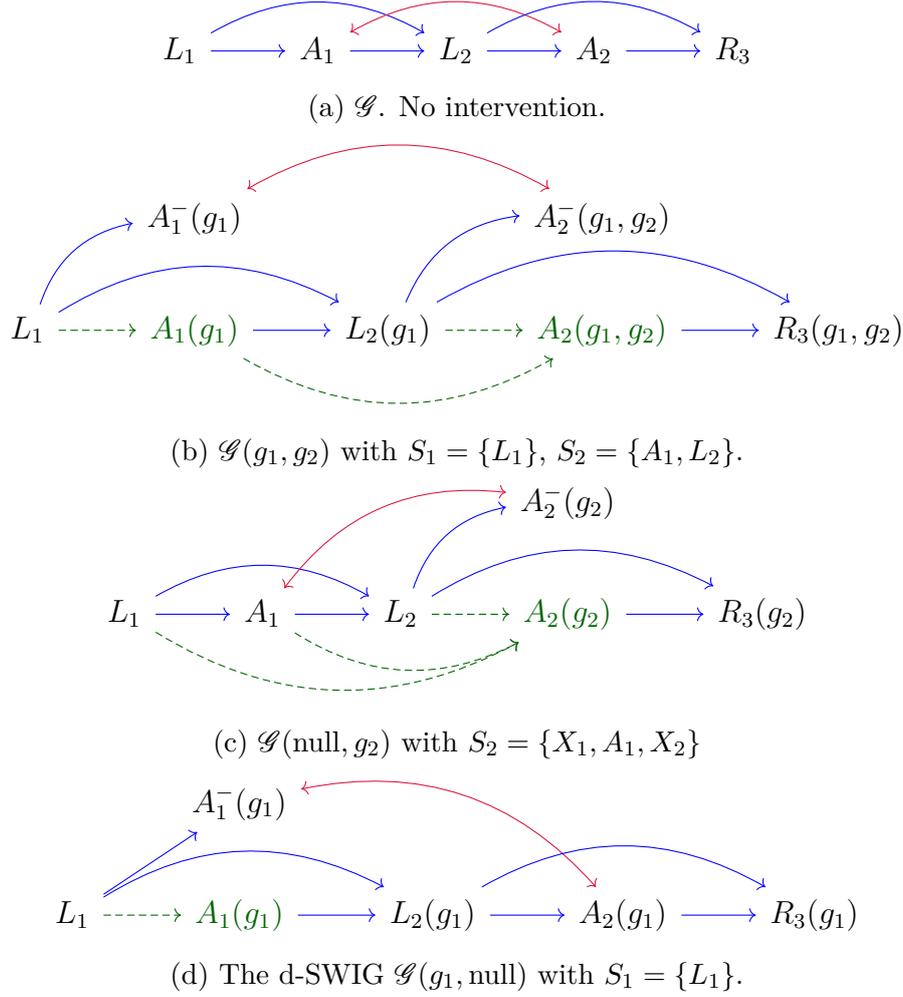

Our causal model as defined above have two useful properties. The
first property highlights the role of the d-SWIG.
\begin{proposition} \label{prop:g-markov}
  The distribution of $(V(g), A^-(g))$ is global Markov with respect
  to $\gG(g)$.
\end{proposition}
Next, we introduce a notion of \emph{consistency} that links
potential outcomes to observed outcomes. The causal inference
literature usually considers static interventions, for which
consistency (e.g.\ $V_j = V_j(V_{\pa(j)})$) naturally arises from the
structural equations. That is, when the observed and
counterfactual values of the parents agree, the observed and potential
outcome should coincide. Because we are considering adaptive
interventions, our consistency notion is slightly more complicated but folows from the same
underlying principle. Denote the future innovations at time $t$ as
$\underline{N}_t = (N_t, N_{t+1}, \dots, N_T)$ and future actions as
$\underline{A}_t = (A_t, A_{t+1}, \dots, A_T)$. Let $\underline{g}_t =
(\text{null}, \dots, \text{null}, g_t, g_{t+1}, \dots, g_T)$ denote
the ``sub-policy'' that only intervenes at and after time $t$. 

\begin{assumption}[Dynamic consistency] \label{assump:dynamic-consistency}
For any $t \in [T+1]$, we have $V_j(\underline{g}_t) = V_j
$ for every $V_j \prec A_{t}$ and the following recursion
    \begin{equation} \label{eq:dynamic-consistency}
        \begin{split}
            &\PP(\underline{N}_t(\underline{g}_{t-1}) =
              \underline{n}_t, \underline{A}_t(\underline{g}_{t-1}) =
              \underline{a}_t \mid A_{t-1}^{-}(\underline{g}_{t-1}) =
              A_{t-1}(\underline{g}_{t-1}) = a_t,
              S_{t-1}(\underline{g}_{t-1}) = s_{t-1}) \\
            = &\PP(\underline{N}_t(\underline{g}_{t}) =
                \underline{n}_t, \underline{A}_t(\underline{g}_{t}) =
                \underline{a}_t \mid A_{t-1}(\underline{g}_{t}) =
                a_{t-1}, S_{t-1}(\underline{g}_{t}) = s_{t-1}),
        \end{split}
      \end{equation}
      with the convention that $\underline{g}_{T+1}$ is the null
      intervention (so $V(\underline{g}_{T+1}) = V$) and $A_{T+1} =
      \emptyset$.
\end{assumption}

\begin{proposition} \label{prop:dynamic-consistency}
  \Cref{assump:dynamic-consistency} is true in the causal model
  defined above.
\end{proposition}
Intuitively, dynamic consistency says that given
no intervention takes place before time $t-1$ and the natural
counterfactual and policy intervention values of the decision
$A_{t-1}$ agree, we can ignore the policy intervention at time $t-1$
and treat future data as being generated from
$\PP(\underline{g}_t)$. 

\section{Identification of the value of a policy}\label{sec:identication}
In this section we discuss our recursive identification strategy in more detail
and prove \Cref{thm:main_result}. 
Let us first introduce a milder unconfoundedness
assumption that is sufficient for proving the main theorem.
\begin{assumption*}{\ref{assump:no-confounding-a}*}[Dynamic unconfoundedness] \label{assump:no-confounding-b}
    The following m-separations are true:
    \begin{align*} 
        \textnot N_{t+1}(\underline{g}_{k+1}) \confpath A_k \mid S_k
      \ingraph{\gG(\underline{g}_{k+1})}, ~ k \leq t \leq T.
    \end{align*}
\end{assumption*}
\begin{proposition}\label{prop:strong-unconfoundedness-implies-weak}
    Under Assumption \ref{assump:nested-states}, Assumption \ref{assump:no-confounding-a} implies Assumption \ref{assump:no-confounding-b}.
\end{proposition}

As will be seen in \Cref{sec:remarks-dynam-unconf} below, 
Assumption \ref{assump:no-confounding-a} is not much stronger than
Assumption \ref{assump:no-confounding-b}. We choose to present
Assumption \ref{assump:no-confounding-a} in the Introduction
because it is a condition on the original causal graph $\gG$ and is
thus easier to verify.

We now present two key conditional independences that follow from our
causal model and assumptions. Both results follow immediately from
applying \Cref{prop:g-markov} to the corresponding m-separations;
see the Online Supplement.
\begin{lemma}\label{lemma:innovation_action_unconfoundedness_gt}
    Assume that $R_t \subseteq S_t$ for $t \in [T]$. Under \Cref{assump:nested-states} and \ref{assump:no-confounding-b}, we have
    \begin{align*}
        (\underline{N}_{t+1} \cup
      \underline{A}_{t+1})(\underline{g}_t) \indep
      A^{-}_t(\underline{g}_t) \mid (A_t \cup S_t)(\underline{g}_t),~t
      \in [T].
    \end{align*}
\end{lemma}
Intuitively, \Cref{lemma:innovation_action_unconfoundedness_gt} states
that the future potential outcomes of innovations and actions are
independent of the actual treatment at time $t$ given the current
state and action. This is commonly referred to as \emph{no
  confoundedness} or \emph{ignorability} in the causal inference
literature. We will compare our version in more detail to the
well-known assumption of \emph{sequential ignorability} from the DTR
literature in \Cref{sec:dtr}.
\begin{lemma}\label{lemma:innovation_action_memorylessness_gt}
    Assume that $R_t \subseteq S_t$ for $t \in [T]$. Under Assumptions \ref{assump:nested-states} and \ref{assump:memorylessness}, we have
    \begin{align*}
        (\underline{N}_{t+1} \cup
      \underline{A}_{t+1})(\underline{g}_t) \indep (A_{t-1} \cup
      S_{t-1}) \setminus S_t \mid (A_t \cup
      S_t)(\underline{g}_t),~\text{for all}~t \in [T].
    \end{align*}
\end{lemma}
\Cref{lemma:innovation_action_memorylessness_gt} states that the
future potential outcomes of innovations and actions are independent
of the history given the current state and action. In the MDP
literature, this is referred to as the \emph{Markov property}, which
requires that the future is independent of the past given the
present. Note that in our setting, there may be other variables in
$V(\underline{g}_t)$ that are not
independent of the future innovations and actions, but we only require
independence from the history of previous actions and
states. 

Finally, before proving our main result, we need a technical but
necessary assumption that ensures a non-zero probability of treatment
given the current state. This is commonly referred to as
\emph{positivity} \parencite{hernan_causal_nodate}. This is a standard
assumption in the causal inference literature, and thus is
implicit in the Introduction. 
\begin{assumption}[Positivity]\label{assump:positivity}
  Given the state $S_t$, the probability density of $A_t$
  under $\PP$ is strictly positive, that is,
  $
    \PP(A_{t} = a_t \mid S_{t} = s_t) > 0 ~ \text{for all $a_t, s_t$
      and $t \in [T]$}
   $.
\end{assumption}

We are now ready to state the main result in this section.
\begin{theorem}\label{thm:innovation_recursion}
    Assume that $R_t \subseteq S_t$ for $t \in [T]$. Under Assumptions
    \ref{assump:nested-states}, \ref{assump:memorylessness},
    \ref{assump:no-confounding-b}, \ref{assump:dynamic-consistency},
    and \ref{assump:positivity}, we have the following recursion for
    all $t \in [T+1]$:  
    \begin{equation}
      \label{eq:innovation_recursion}
      \begin{split}
        &\mathbb{P}(\underline{N}_{t}(\underline{g}_{t-1}) = \underline{n}_{t}, \underline{A}_{t}(\underline{g}_{t-1}) = \underline{a}_{t} \mid A_{t-1}(\underline{g}_{t-1}) = a_{t-1}, S_{t-1} = s_{t-1}) \\
        = &\mathbb{P}(\underline{N}_{t+1}(\underline{g}_t) = \underline{n}_{t+1}, \underline{A}_{t+1}(\underline{g}_t) = \underline{a}_{t+1} \mid A_{t}(\underline{g}_t) = a_t, S_t = s_t) g_t(a_t \mid s_t)\PP(n_t \mid a_{t-1}, s_{t-1}),
      \end{split}
    \end{equation}
    in which we use the convention that $A_0$, $S_0$,
    $N_{T+2}$, $S_{T+1}$, $A_{T+1}$, $A_{T+2}$ are empty.
  \end{theorem}

\begin{proof}[Proof of
  \Cref{thm:main_result,thm:innovation_recursion}]
  \Cref{thm:innovation_recursion} follows from a sequence of equalities:
    \begin{align*}
        & \mathbb{P}(\underline{N}_{t}(\underline{g}_{t-1}) = \underline{n}_{t}, \underline{A}_{t}(\underline{g}_{t-1}) = \underline{a}_{t} \mid A_{t-1}(\underline{g}_{t-1}) = a_{t-1}, S_{t-1} = s_{t-1})\\
        =& \mathbb{P}(\underline{N}_{t}(\underline{g}_{t-1}) = \underline{n}_{t}, \underline{A}_{t}(\underline{g}_{t-1}) = \underline{a}_{t} \mid A_{t-1}^{-}(\underline{g}_{t-1}) = A_{t-1}(\underline{g}_{t-1}) = a_{t-1}, S_{t-1} = s_{t-1}) \\
        =& \mathbb{P}(\underline{N}_{t}(\underline{g}_t) = \underline{n}_{t}, \underline{A}_{t}(\underline{g}_t) = \underline{a}_{t} \mid A_{t-1} = a_{t-1}, S_{t-1} = s_{t-1}) \\
        =& \mathbb{P}(\underline{N}_{t+1}(\underline{g}_{t}) = \underline{n}_{t+1}, \underline{A}_{t+1}(\underline{g}_{t}) = \underline{a}_{t+1} \mid A_{t}(\underline{g}_t) = a_t, S_t = s_t, (A_{t-1} \cup S_{t-1}) \setminus S_t = (a_{t-1} \cup s_{t-1}) \setminus s_t)\\
        &g_t(a_t \mid s_t)\PP(n_t \mid a_{t-1}, s_{t-1}) \\
        =& \mathbb{P}(\underline{N}_{t+1}(\underline{g}_t) = \underline{n}_{t+1}, \underline{A}_{t+1}(\underline{g}_t) = \underline{a}_{t+1} \mid A_{t}(\underline{g}_t) = a_t, S_t = s_t) g_t(a_t \mid s_t)\PP(n_t \mid a_{t-1}, s_{t-1}).
    \end{align*}
    The first equality uses
    \Cref{lemma:innovation_action_unconfoundedness_gt} and positivity,
    the second uses dynamic
    consistency, the third equality factorizes the density 
    and uses $R_t \subseteq S_t$, and the fourth uses
    \Cref{lemma:innovation_action_memorylessness_gt}.

\Cref{thm:main_result} then follows from recursively applying
\Cref{thm:innovation_recursion} by noting that the left hand side of the
identification formula in \Cref{thm:main_result} is exactly the
left hand side of \eqref{eq:innovation_recursion} when $t = 1$. Recall
that $N_{T+1} = R_{T+1}$, so for $t = T+1$, equation
\eqref{eq:innovation_recursion} is simply
\begin{align*}
    \mathbb{P}(R_{T+1}(\underline{g}_T) = r_{T+1} \mid A_T(\underline{g}_{T}) = a_T, S_{T} = s_{T}) = \mathbb{P}(R_{T+1} = r_{T+1} \mid A_{T} = a_T, S_T = s_T).
\end{align*}
This is basically the usual identification formula for static
interventions in the causal inference literature given
unconfoundedness.
  \end{proof}

\subsection{Different versions of dynamic
  unconfoundedness}\label{sec:remarks-dynam-unconf}
As \Cref{prop:strong-unconfoundedness-implies-weak} above shows, Assumption
\ref{assump:no-confounding-b} is implied by the ``no dynamic
back-door'' criterion in Assumption
\ref{assump:no-confounding-a}. \Cref{fig:d-SWIG-examples} presents an
example where causal identification is still possible when Assumption
\ref{assump:no-confounding-b} holds but Assumption
\ref{assump:no-confounding-a} does not.
Clearly in \Cref{subfig:d_SWIG_a},
\Cref{assump:no-confounding-a} fails regardless of how $S_1$ is chosen
because of the back-door path
\[
    N_3 \ldedge A_2 \bdedge A_1 \mid S_1 \ingraph{\gG}.
\]
This failure shows that we cannot
identify the value of \emph{all} policies $g$. To see this,
let $g = (g_1, \text{null})$ be a policy that intervenes at time $t=1$
only, and let $S_1 = \{L_1\}$. The corresponding d-SWIG is presented
in
\Cref{fig:d-SWIG-g1}.
Because of the bidirected edge $A_1 \bdedge A_2$ in
\Cref{subfig:d_SWIG_a}, there still exists a back-door path
\[
    R_3(g_1) \ldedge A_2(g_1) \bdedge A_1^{-}(g_1) \mid S_1(g_1) \ingraph{\gG(g_1,\text{null})}
\]
no matter how $S_1$ is chosen. In other words, it is not possible to
choose a state $S_1$ so that $R_3(g_1)$ and $A_1^{-}(g_1)$ are
guaranteed to be conditionally
independent (i.e.\ \Cref{lemma:innovation_action_unconfoundedness_gt}
is not true). Thus, it is generally not possible to identify the
distribution of $R_3(g_1)$.

Surprisingly, it is still possible to identify the
distribution of $R_3(g)$ for $g = (g_1,g_2)$ that intervenes
in both periods
or $g = (\text{null}, g_2)$ that intervenes in the second period
only. As
demonstrated in \Cref{subfig:d_SWIG_c,subfig:d_SWIG_b}, the
intervention on $A_2$ breaks the above back-door path and thus the
confounding dependence between $A_1$ and $R_3$, so Assumption
\ref{assump:no-confounding-b} holds in these cases. 

To provide intuition on this somewhat peculiar example, we
find it useful to draw a parallel to the principle of
optimality in \textcite[Chap.\ III.\ 3]{bellman_dynamic_1957} for
policy learning:
\begin{quote}
    ``\emph{An optimal policy has the property that whatever the
      initial state and initial decision are, the remaining decisions
      must constitute an optimal policy with regard to the state
      resulting from the first decision.}''
\end{quote}
Although we are considering causal identification instead
of policy optimization, a similar reasoning is applied in our proof: in order
to identify the value of a policy $g = (g_1,
\underline{g}_2)$, we just need to make sure that $A_1$ is not
confounded given that we may intervene on all future decisions, so the
problem can be reduced to identify the value of the policy
$\underline{g}_2$. This argument is then applied recursively to prove
\Cref{thm:main_result} using \Cref{thm:innovation_recursion}.



\section{Dynamic Treatment Regimes}\label{sec:dtr}
In the statistical literature on dynamic treatment regimes (DTRs), often just
a single reward $R_{T+1} \in X_{T+1}$ is considered and the state variables include the
entire history.

\begin{assumption}[DTR] \label{assump:DTR}
  We assume $R_1 = \dots = R_T = \emptyset$ and
\begin{equation*} 
  S_t = S_{t-1} \cup A_{t-1} \cup X_t = X_1 \cup A_1 \cup \dots \cup X_{t-1} \cup
  A_{t-1} \cup X_t,~t \in [T].
\end{equation*}
\end{assumption}
It is easy to see that
\Cref{assump:nested-states,assump:memorylessness} are trivially
satisfied under \Cref{assump:DTR}. 
    To ensure that the distribution of $R_{T+1}(g)$ is
identified, it is typically assumed that the treatment assignments satisfy
``sequential ignorability''
\parencite{robins_new_1986,robins_causal_1997,murphy_optimal_2003}. In
our notation, this can be expressed as the following graphical
condition. 
\begin{assumption}[Sequential
  ignorability] \label{assump:dtr_seq_ignorability} We have the
  following m-separations :
\[    \textnot    R_{T+1}(g) \mconn A_t^{-}(g) \mid S_t(g) \cup
\overline{A}^{-}_{t-1}(g) \ingraph{\gG(g)}~\text{for all $t \in
  [T]$}.\] 
  \end{assumption}

The next proposition shows that this graphical criterion is equivalent
to our dynamic unconfoundedness in Assumption \ref{assump:no-confounding-b}
if we make a mild assumption that there exists a directed path from
every state variable at any time point to the final reward:
\begin{equation}
  \label{eq:directed-path-state-reward}
        V_j(g) \rdpath R_{T+1}(g) \ingraph{\gG(g)}, ~ \text{for
          all}~V_j \in \overline{S}_T.
\end{equation}

\begin{proposition}\label{prop:seq_ignorability_and_dyn_unconfoundedness}
    Let \Cref{assump:DTR} be given. Then
    Assumption \ref{assump:no-confounding-b} implies
    \Cref{assump:dtr_seq_ignorability}. Conversely, if we further
    assume \eqref{eq:directed-path-state-reward}, then
    \Cref{assump:dtr_seq_ignorability} implies
    Assumption \ref{assump:no-confounding-b}.
\end{proposition}

Thus, our setup generalizes the DTR setting because \Cref{assump:DTR}
is not required. This can be seen from the following quick proof
of the well-known g-computation formula in
\textcite{robins_new_1986,robins_causal_1997}.
\begin{corollary}[Robins' g-formula] \label{cor:g-formula}
  Let dynamic consistency and positivity be given. Under
  Assumptions \ref{assump:DTR} and \ref{assump:dtr_seq_ignorability} and
  \eqref{eq:directed-path-state-reward}, the joint distribution under
  policy $g$ is identified by
  \begin{align*}
    &\mathbb{P}(R_{T+1}(g) = r_{T+1}, \overline{X}_{T}(g) =
      \overline{x}_{T}, \overline{A}_{T}(g) = \overline{a}_{T}) \\
    =& \mathbb{P}(r_{T+1} \mid \overline{x}_T, \overline{a}_T) \prod_{t=1}^T g(a_{t} \mid \overline{x}_{t}, \overline{a}_{t-1}) \PP(x_{t} \mid \overline{x}_{t-1}, \overline{a}_{t-1}).
  \end{align*}
\end{corollary}
\begin{proof}
  By \Cref{prop:seq_ignorability_and_dyn_unconfoundedness}, Assumption
  \ref{assump:no-confounding-b} is satisfied. This result then follows
  from recursively applying \Cref{thm:innovation_recursion} in the
  same way as in our proof of \Cref{thm:main_result}.
\end{proof}

Compared to \textcite[Corollary 34]{richardson_single_2013} who aim
to identify the distribution of $R_{T+1}(g)$, here we
assume \eqref{eq:directed-path-state-reward} in addition. This is
needed to identify the joint distribution of all variables and is not
a strong assumption: \eqref{eq:directed-path-state-reward} is true
when every non-decision variable is used by the policy for at least
one decision (so $V_j(g) \rdedge A_t(g)$ for some $t$ and all $V_j \in
X$) and every decision has a causal effect on the final reward (so
$A_t(g) \rdpath R_{T+1}(g)$ for all $t$). This assumption can be
avoided if we are not interested in identifying the non-decision
variables that have no causal effect on $R_{T+1}$.



\section{Markov decision processes}\label{sec:mdp}
In RL, the data generating mechanism, often referred to as the
``environment'', is typically stated in the form of an MDP, where a
reward $R_t$ is observed at every time-point
\parencite{sutton_reinforcement_2018,puterman_markov_2014}.
A large body of the RL literature focuses on estimating the value of an
\emph{evaluation policy} from data generated using another policy.
This problem is known as \emph{offline policy
  evaluation}
\parencite{sutton_reinforcement_2018,uehara_review_2022}. However, no
formal causal theory exists for MDPs to the
best of our knowledge. Thus, it is not
always clear what causal assumptions are made behind the scenes in
off-policy evaluation learning in RL; this is discussed in
\Cref{sec:mdp-assumptions}. In subsequent subsections, we discuss the
cross-temporal nature of our Assumption
\ref{assump:no-confounding-a}, review the use of causal diagrams
in the MDP literature, and discuss why a popular
relaxation---\emph{partially observed Markov decision process}
(POMDP)---is generally not identified.

\subsection{Additional assumptions in the MDP literature} \label{sec:mdp-assumptions}

\subsubsection{Time-invariance}
While time-invariance (sometimes referred to as \emph{stationarity} or
\emph{time-homogeneity}) does not have any implications for the
identification problem, it is a standard assumption in the MDP
literature. We state it here for completeness.
\begin{assumption}[Time-invariant reward and state
  transitions]\label{assump:mdp_time_invariance}
  The conditional distribution
  \begin{equation}
    \label{eq:transition}
    \PP(S_{t+1}(g) = s', R_{t+1}(g) = r \mid A_t(g) = a, S_t(g) =
    s)
  \end{equation}
  is independent of $t$ and $g$.
\end{assumption}

Because the policy $g$ does not directly intervenes on the state or
the reward, the requirement that \eqref{eq:transition} is independent
of $g$ is natural from a causal perspective. The real assumption is
that it is also independent of time, which implies that the state
variables $S_1, S_2, \dots$ take values in the same space. In the RL
literature, the expression in \eqref{eq:transition} is often
denoted as $\PP(s', r \mid a, s)$ or something similar
\parencite{sutton_reinforcement_2018,uehara_review_2022}. The
conditional probability distribution of the next state $\PP(S_{t+1}(g)
= s' \mid A_t(g) = a, S_t(g) = s)$, obtained by marginalising
\eqref{eq:transition}, is usually referred to as the \emph{transition
  probability function}. In
addition to Assumption \ref{assump:mdp_time_invariance}, some authors
assume that the next state and reward are statistically
independent given the state and action, i.e.\ $\PP(s', r
\mid a, s) = \PP(s' \mid a, s) \PP(r \mid a, s)$. Such assumptions
have implications for computating the tangent spaces in
semiparametric settings \parencite{kallus_efficiently_2021}.

\begin{assumption}[Time-invariant policy] \label{assump:policy_time_invariance}
  The evaluation and null policies are time-invariant (and thus only depends
  on the state) in the sense that
  $\PP(A_t(g) = a \mid S_t(g) = s)$ is independent of $t$.
\end{assumption}

This assumption enables importance sampling for
off-policy evaluation when the null policy needs to be estimated from
a single trajectory. With
\Cref{assump:mdp_time_invariance,assump:policy_time_invariance}, an
MDP is commonly described as a tuple of state space, action space, state
transition probability function, and reward function (expected reward
given state and action).




\subsubsection{Randomized decisions}
RL has in large part been popularized through impressive results in
games \parencite{silver_mastering_2016}. State selection is often an
easy task in such settings. For example, in a chess game the
current board position contains all the information
needed for a player to make the next move. More broadly, it is
common practice to benchmark RL methodologies
against settings where the state is \emph{prespecified} as part of the
environment. \texttt{Gymnasium}\footnote{\texttt{Gymnasium} is a
  maintained fork of
  OpenAI's Gym library. See \url{https://gymnasium.farama.org/}.} is a
popular API often referenced in scientific work that offers a range
of RL environments. For example, in the CartPole environment in
\texttt{Gymnasium}, the
task is to balance a pole on a cart moving along a frictionless track
\parencite{bartosutton1983}. The prespecifie state is cart
position, cart velocity, pole angle, and pole angular velocity.
In the MountainCar environment, one
must accelerate a car up a hill of a sinusoidal valley to gain
momentum \parencite{Moore90efficientmemory-based}. The prespecified
state is position (x-axis) and velocity of the car.


Clearly, when the state $S_t$ is prespecified by the environment, any
natural policy used to generate
historical data is \emph{unconfounded} given the state. Consequently,
the data may be viewed as the result of a \emph{stratified randomized
  experiment}. This can be formally stated as the following graphical
requirement.
\begin{assumption}[Randomized
  decisions]\label{assump:rl_behavior_policy}
  In the causal graph $\gG$, there is no bidirected edge with one end
  being a decision variable $A_t$, and the state $S_t$ contains the
  parent set of $A_t$ for all $t \in [T]$.
\end{assumption}

\Cref{assump:rl_behavior_policy} clearly implies
\Cref{assump:no-confounding-a}, because any backdoor path from $A_k$
must begin with $A_k \ldedge S_k$. Thus, much of the RL literature
makes the causal assumption of dynamic unconfoundedness without
an explicit statement. Rather than using a causal diagram, many
authors uses the factorization
\[
  \PP(\underline{a}_1,\underline{s}_1) = \PP(s_1) \prod_{t=1}^T
  \PP(a_t \mid s_t) \PP(s_{t+1}, r_{t+1} \mid a_t, s_t)
\]
and makes an implicit assumption that under a policy $g = (g_1,\dots,
g_T)$, the probability $\PP(g)$ factorizes in the same way with the term
$\PP(a_t \mid s_t)$ replaced by $g_t(a_t \mid s_t)$.


Although \Cref{assump:rl_behavior_policy} may be quite reasonable in
simulated environments, it may not hold in real-world problems. In the
causal ADMG, \Cref{assump:rl_behavior_policy} fails if the graph contains any
edge of the kind $A_t \fullstraighalf V \setminus S_t$. This is the
case in dynamic pricing problem depicted in
\Cref{sub@subfig:competitor_price},
where a shipping professional relies on verbally communicated sales
intelligence to set prices. As another example, when
prescribing a drug, a physician might consider the patient's own
concerns that are not recorded in the health database. In RL, the null
policy that generates the observed data is often referred to as the
behavior policy. In light of the discussion above, this
terminology can be misleading because a real-world ``behavior'' may not
satisfy \Cref{assump:rl_behavior_policy}.



\subsection{Why controlling for immediate confounding is not enough}
For readers with a background in RL, it may seem odd at first why
Assumption \ref{assump:no-confounding-a} requires blocking back-door
paths across all time periods, especially given the memorylessness in
Assumption \ref{assump:memorylessness}. In other words, one might
wonder why
\begin{equation}\label{eq:dynamic_back_door_t_only}
    \textnot N_{t+1} \confpath A_t \mid S_t \ingraph{\gG} ~ \text{for all $t \in [T]$},
\end{equation}
is not sufficient for identifiability. The reason is that when $S_t$
includes previous actions, blocking the paths in
\eqref{eq:dynamic_back_door_t_only} may not be enough.
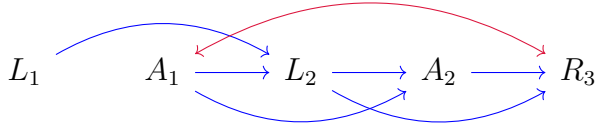
\begin{figure}[t]
  \centering
  \begin{tikzcd}
      L_1 \arrow[rr, graphBlue, bend left] & A_1 \arrow[r, graphBlue] \arrow[rr, graphBlue, bend right] \arrow[rrr, graphRed, leftrightarrow, bend left] & L_2 \arrow[r, graphBlue] \arrow[rr, graphBlue, bend right] & A_2 \arrow[r, graphBlue] & R_3
  \end{tikzcd}
  \caption{Memorylessness and \eqref{eq:dynamic_back_door_t_only} are
    satisfied for $S_2 = \{A_1, L_2\}$, but
    \Cref{assump:no-confounding-a} fails.}
  \label{fig:d-SWIG-2a-explained}
\end{figure}
\Cref{fig:d-SWIG-2a-explained} gives an example. Clearly,
memorylessness is satisfied when $S_2 = \{A_1, L_2\}$, that is, we
have
\[
    \textnot R_3 \mconn (L_1 \cup A_1) \setminus S_2 \mid S_2.
\]
Furthermore, $S_2 = \{A_1, L_2\}$ ancestrally blocks all back-door
paths from $A_2$ to $R_3$, so \eqref{eq:dynamic_back_door_t_only} is
also satisfied. However, the bidirected edge $A_1 \bdedge R_3$ clearly
precludes causal identification of any policy with any choice of
state. The existence of this
bidirected edge fails \Cref{assump:no-confounding-a}, which requires
all back-door paths from $A_1$ to $R_3$ to be ancestrally blocked.

Despite this observation, the next proposition shows that
\eqref{eq:dynamic_back_door_t_only} is sufficient when $A_{t-1} \notin
S_t$, that is, if previous decisions are not included in the
states. This may be reasonable in some games and simulated
environments but difficult to defend in some other problems.

\begin{proposition}\label{prop:dynamic_back_door_no_At_in_St}
  Assume $A_{t-1} \notin S_t$ for all $t \in [T]$ and Assumptions
  \ref{assump:nested-states} and \ref{assump:memorylessness}. Then
  \Cref{assump:no-confounding-a} is equivalent to \eqref{eq:dynamic_back_door_t_only}.
\end{proposition}

\subsection{The abuse of causal diagrams in the MDP literature}
In light of our results, we next review the use of causal diagrams
in the RL literature \parencite[see
e.g.][]{kallus2020double}. Typically, it is assumed that the MDP is
time-invariant
(\Cref{assump:mdp_time_invariance,assump:policy_time_invariance} are
satisfied), and
the states, actions and rewards are represented as vertices in a
DAG. This is illustrated in \Cref{fig:rl_causal_diagram} that repeat
across two time points. Some authors also include
latent variables and edges like $A_t \ldedge U_t \rdedge S_{t+1}$ in
the DAG to indicate unmeasured confounding
\parencite{xu2023instrumentalvariableapproachconfounded}.
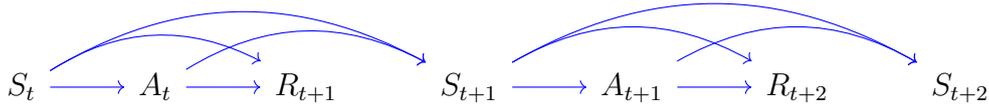
\begin{figure}[t]
  \begin{subfigure}[b]{\textwidth}
      \centering
      \begin{tikzcd}
          S_t \arrow[r, graphBlue] \arrow[rr, graphBlue, bend left] \arrow[rrr, graphBlue, bend left] & A_t \arrow[r, graphBlue] \arrow[rr, graphBlue, bend left] & R_{t+1}  & S_{t+1} \arrow[r, graphBlue] \arrow[rr, graphBlue, bend left] \arrow[rrr, graphBlue, bend left] & A_{t+1} \arrow[r, graphBlue] \arrow[rr, graphBlue, bend left] & R_{t+2} & S_{t+2}
      \end{tikzcd}
  \end{subfigure}
  \caption{A Markov decision process represented by a DAG $\gG$.}

  \label{fig:rl_causal_diagram}
\end{figure}

From a causal perspective, diagrams like \Cref{fig:rl_causal_diagram}
should be interpreted as the latent projection of the original graph
$\gG$ \parencite{verma_pearl_1990} onto all the states, actions, and
rewards. However, this latent projection can be misleading when the
MDP constructs---the states, actions, and rewards---have overlapping
variables. For this reason, we
recommend against using causal diagrams in which the vertices are states,
decisions, and rewards (such as
\Cref{fig:rl_causal_diagram,fig:wrong_projection}).

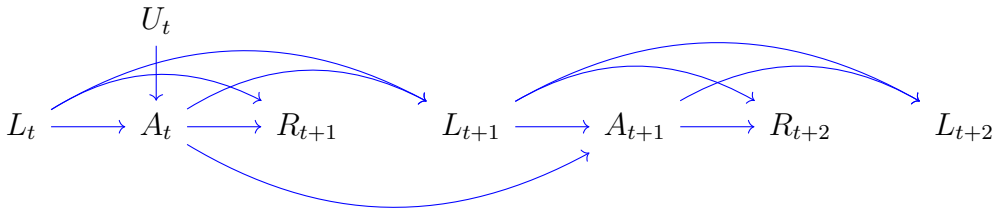
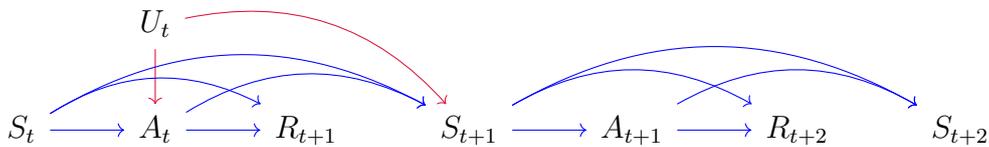
\begin{figure}[t]
  \centering
  \begin{subfigure}[b]{\textwidth}
      \centering
    \begin{tikzcd}
      & U_t \arrow[d, graphBlue] &&&&& \\
      L_t \arrow[r, graphBlue] \arrow[rr, graphBlue, bend left] \arrow[rrr, graphBlue, bend left] & A_t \arrow[r, graphBlue] \arrow[rr, graphBlue, bend left] \arrow[rrr, graphBlue, bend right] & R_{t+1}  & L_{t+1} \arrow[r, graphBlue] \arrow[rr, graphBlue, bend left] \arrow[rrr, graphBlue, bend left] & A_{t+1} \arrow[r, graphBlue] \arrow[rr, graphBlue, bend left]  & R_{t+2} & L_{t+2}
    \end{tikzcd}
    \caption{Original causal graph.}
  \label{fig:rl_original_graph}
  \end{subfigure}

  \begin{subfigure}[b]{\textwidth}
      \centering
      \begin{tikzcd}
        & U_t \arrow[d, graphRed] \arrow[rrd, graphRed, bend left] &&&& \\
        S_t \arrow[r, graphBlue] \arrow[rr, graphBlue, bend left] \arrow[rrr, graphBlue, bend left] & A_t \arrow[r, graphBlue] \arrow[rr, graphBlue, bend left] & R_{t+1}  & S_{t+1} \arrow[r, graphBlue] \arrow[rr, graphBlue, bend left] \arrow[rrr, graphBlue, bend left] & A_{t+1} \arrow[r, graphBlue] \arrow[rr, graphBlue, bend left] & R_{t+2} & S_{t+2}
      \end{tikzcd}
      \caption{Latent projection with the choice of state $S_t =
        \{A_{t-1},L_{t}\}$.}
  \label{fig:wrong_projection}
  \end{subfigure}

  \caption{An illustration of ``phantom confounding'' created by latent projection.}
  \label{fig:phantom}
\end{figure}

To illustrate this, consider the causal graph in
\Cref{fig:rl_original_graph}. If we choose the state as $S_t =
\{L_t\}$, memorylessness is violated due to the directed edge $A_t
\rdedge A_{t+1}$. In contrast, the choice $S_t = \{A_{t-1}, L_t\}$ satisfies
Assumptions \ref{assump:nested-states}, \ref{assump:memorylessness},
and \ref{assump:no-confounding-a}, and thus the value of any policy that
depends on this choice of $S_t$ can be identified by
\Cref{thm:main_result}. However, we would come to a wrong conclusion
about this state choice if we use the latent projection graph as shown
in \Cref{fig:wrong_projection}. Due to the additional edge $U_t \rdedge
S_{t+1}$ (which needs to be added because $A_t \in S_{t+1}$), there is
a ``phantom'' back-door path $A_t \ldedge U_t \rdedge S_{t+1}$ that
cannot be blocked because $U_t$ is not observed.

\subsection{Violation of memorylessness in POMDPs}\label{subsec:pomdps}
Modeling decision processes as MDPs is appealing for many reasons. For
instance, time-invariant MDPs do not suffer from the curse of
dimensionality in long and infinite horizon
problems. 
However, as discussed earlier, MDPs impose strong assumptions
on the data generating mechanism. 

A prominent relaxation of MDPs is the \emph{partially observed Markov
  decision processes} (POMDPs) \parencite{ASTROM1965174,
  KAELBLING199899} that assume imperfect
information. In recent years POMDPs have received considerable
attention as a more reasonable trade-off between model flexibility and
structural assumptions \parencite[e.g.][]{guo2016pac}. In stage $t$
of a POMDP, it is assumed that the decision process can be fully
described by some latent state $H_t$ goverened by a transition kernel
$\PP(u_{t+1} \mid a_t, h_t)$. The agent measures a subset or
transformation $X_t$ of the true state $U_t$ through some probability
distribution $\PP(x_t \mid h_t)$.\footnote{It is sometimes
  assumed that $X_t$ also depends on previous actions.} Usually, it is
assumed that the observed data is randomized (see
\Cref{assump:rl_behavior_policy}) so that $A_t$ only
depends on $X_t$. Finally, the decision-maker received a reward
$R_t$. \Cref{fig:pomdp} shows the typical setting of a POMDP.

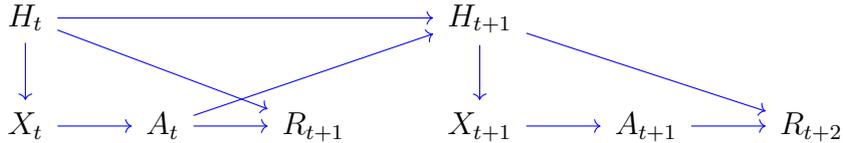
\begin{figure}[t]
  \centering
  \begin{tikzcd}
    H_t \arrow[d, graphBlue] \arrow[rrd, graphBlue] \arrow[rrr, graphBlue] &  &  & H_{t+1} \arrow[d, graphBlue] \arrow[rrd, graphBlue] &   &  \\
    X_t \arrow[r, graphBlue] & A_t \arrow[r, graphBlue] \arrow[rru, graphBlue] & R_{t+1} & X_{t+1} \arrow[r, graphBlue] & A_{t+1} \arrow[r, graphBlue] & R_{t+2}
  \end{tikzcd}
  \caption{A graphical representation of POMDP.}
  \label{fig:pomdp}
\end{figure}

Clearly, the decision process in \Cref{fig:pomdp} satisfies
\Cref{assump:nested-states,assump:memorylessness,assump:no-confounding-a}
with $S_t = \{H_t\}$ or $S_t = \{H_t, X_t\}$. However, this is
infeasible because $H_t$ is unobserved. In fact, it is clear that
memorylessness fails unless $S_t$ includes the entire history. As an
example, there is a path $
  R_{t+2} \ldedge H_{t+1} \ldedge H_t \rdedge X_t$ not ancestrally
  unblocked by $S_{t+1}$ when $S_t = \{X_t\}$ and $S_{t+1} =
  \{X_{t+1}\}$.

For this reason, POMDPs are generally not identified. However,
approximate inference may still be possible without including the full
history, for example if the process is sufficiently mixed
\parencite{hu2023off}.

\section{Simulation Study for Dynamic Pricing}\label{sec:simulation_study}
We now revisit the dynamic pricing problem and expand the
simulation in \Cref{sec:motivating-example}. 

\subsection{Basic simulation setup}

In \Cref{sec:motivating-example}, the container logistics company sets
a single price for a specific departure. In
reality this price is continuously changed until either the total
bookings exceed the vessel's capacity or the departure date is
reached,
and customers may book containers for multiple
departures at any point in time. To capture the dynamic nature of this
pricing problem, we consider a
simulation setting with weekly vessel departures where
the company publishes the initial price two weeks before each
departure and can change the price a week before the departure.

\begin{figure}[ht]
  \centering
  \begin{tikzcd}[nodes={text width=1cm, align=center}]
    & D_{t-1} \arrow[rr, graphOrange, dashed] \arrow[d, graphBlue] \arrow[rdddd, graphBlue, bend left=40] \arrow[rrrdddd, graphBlue, bend left=20] &  & D_t \arrow[d, graphBlue] \arrow[rddddd, graphBlue, bend left=40] & & \\
    \cdots & \hat{D}_{t-1}  \arrow[ddd, graphBlue, bend right=40] &  & \hat{D}_t \arrow[dddd, graphBlue, bend right=30] & \\ \\
    B_{t-1}^1 \arrow[r, graphBlue] \arrow[rd, graphBlue] \arrow[rr,
    graphBlue, rightarrow, bend left] & A_{t-1}^2  \arrow[r,
    graphBlue] \arrow[rrd, graphGreen, dashed] \arrow[rrrd,
    graphGreen, dashed] & B_{t}^2 &  & & \cdots \\
    & A_{t-1}^1 \arrow[r, leftrightarrow, graphRed, bend left, dashed]
    \arrow[rrd, graphBlue] \arrow[r, graphBlue] \arrow[rr, graphBlue,
    bend left] & B_{t}^1 \arrow[r, graphBlue] \arrow[rd, graphBlue]
    \arrow[rr, graphBlue, rightarrow, bend left] & A_{t}^2 \arrow[r,
    graphBlue] \arrow[rrd, graphGreen, dashed] & B_{t+1}^2 \\
    && & A_{t}^1 \arrow[r, graphBlue] \arrow[rr, graphBlue, bend left]
    \arrow[r, leftrightarrow, graphRed, bend left, dashed] \arrow[rrd,
    graphBlue]  & B_{t+1}^1 \arrow[r, graphBlue] \arrow[rd, graphBlue]
    & A_{t+1}^2 & \\
    &&& & & A_{t+1}^1 & \\
    &&& & &
  \end{tikzcd}
  \caption{Dynamic pricing ADMG. The basic scenario $\gG$ consists of all
    solid blue edges. Given the state set $S_t = \{A^1_{t-1},
    \hat{D}_{t-1}, B_{t}^1, \hat{D}_{t}\}$, the orange arrow $D_{t-1}
    ~ \textcolor{graphOrange}{\longrightarrow} ~ D_t$ signifies a
    violation of the Markov property, and $A^2_{t+1} ~
    \textcolor{graphGreen}{\longleftarrow} ~ A^2_t ~
    \textcolor{graphGreen}{\longrightarrow} ~ B^2_{t+2}$ signifies a
    violation of no dynamic back-door. Finally, $A^1_t ~
    \textcolor{graphRed}{\longleftrightarrow} ~ B^1_{t+1}$ denotes
    unmeasured confounding that cannot be controlled for.}

  \label{fig:pricing_example}
\end{figure}
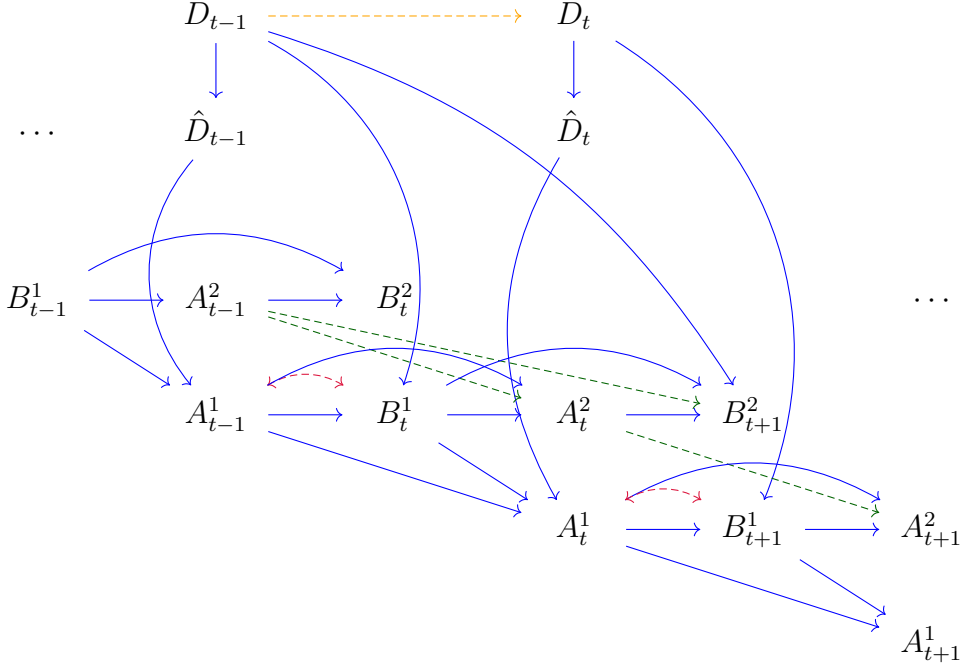

To set up the problem, we use subscripts to index calendar time and
superscripts to index relative time for a departure. For the
vessel departing at time $t+1$, the initial price $A_{t-1}^1$ is
published two weeks before the departure. In the following week,
customers book container slots $B_t^1$ at this price. The price is
then updated to $A^2_t$, and in the second week customers book slots
$B^2_{t+1}$ at the updated price. Looking in a different way, at each
time $t$ we observe four variables: $B_t^2$
is the second-week booking for the vessel departing at time $t$,
$B_t^1$ is the first-week booking for the vessel departing at time
$t+1$, $A_t^2$ is the updated price for the vessel departing at
time $t+1$, and $A_t^1$ is the initial price for the vessel departing
at time $t+2$. The revenue for the week is then $R_t = A^2_{t-1}
B^2_{t} + A^1_{t-1} B^1_{t}$.

The exact structural assumptions between the prices and bookings are
depicted in the ADMG in \Cref{fig:pricing_example} and we will
consider several scenarios. Let
$\gG$ be graph for the basic scenario consisting of the blue arrows
($\textcolor{graphBlue}{\longrightarrow}$) in
\Cref{fig:pricing_example}. We assume the initial price
$A^1_t$ is determined by shipping professionals using the previous
initial price $A^1_{t-1}$, its associated bookings $B^1_t$, and an
estimate $\hat{D}_t$ of the true market demand $D_t$ (for departure at
time $t+2$ and $D_t$ cannot be observed).\footnote{The noisy estimate
  $\hat{D}_t$ may reflect all
the information of the latent demand the decision-maker has access
to. For example, the Chinese New Year heavily impacts the demand for
cargo flowing between East Asia and Western Europe. $\hat{D}_t$ may
also reflect intel from forecast reports that shipping professional
rely on.} In the basic scenario, the
price is updated to
$A^2_{t+1}$ in the second week based on the initial price $A^1_t$ and
the associated bookings $B^1_{t+1}$. The initial bookings $B^1_{t+1}$ is
determined by the initial price $A^1_t$ and latent market demand
$D_t$. In the second period the bookings $B^2_{t+2}$ is determined
by the updated price $A^2_{t+1}$, the initial bookings $B^1_{t+1}$
(through capacity constraints), and the latent market demand
$D_t$. The reward variable $R_{t}$ is not shown in
\Cref{fig:pricing_example} to simplify the graph.

In this basic scenario (with only blue arrows in
\Cref{fig:pricing_example}), the reader is invited to verify that the
state set
\begin{equation}\label{eq:dynamic_pricing_state}
  S_t = \{A^1_{t-1}, \hat{D}_{t-1}, B^1_t, \hat{D}_t\}
\end{equation}
satisfies Assumptions \ref{assump:nested-states},
\ref{assump:memorylessness}, and \ref{assump:no-confounding-a}.



As is common in the economics literature, we assume that the
unconstrained bookings is Poisson distributed with mean linearly determined by
the parents in the graph: a high latent demand $D_t$ for a departure
means the price-elasticity is lower. We assume a fixed capacity $C$
for every vessel so that $B^1_{t+1} + B^2_{t+2} \leq C$. 
For the null policy, the initial and updated prices are piecewise
linear
functions of their parents (and with a small probability the prices are
uniformly distributed to ensure positivity). 
More details on the simulation setup can be found in the Online
Supplement.

\subsection{More sophiscated scenarios}

We compare the performance of the null policy to the policies learned
from policy iteration using different state sets in the basic and
three more sophisticated scenarios:
\begin{description}
  \item[1. Macroeconomic trend.] Suppose there is some macroeconomic
    trend causing inertia in the latent demand over time. Graphically,
    this can be modelled by adding $D_{t-1} ~
    \textcolor{graphOrange}{\longrightarrow} ~ D_t$ to $\gG$ in
    \Cref{fig:pricing_example}. We denote the resulting graph as
    $\overset{\textcolor{graphOrange}{\longrightarrow}}{\gG_1}$. The
    degree of inertia is modeled by varying the probability of
    $D_{t} = D_{t-1}$ in the simulation. Due to this inertia, the
    state choice in \eqref{eq:dynamic_pricing_state} no longer
    satisfies memorylessness (Assumption
    \ref*{assump:memorylessness}); for example, there is an unblocked path
    $B^1_{t+1} \leftarrow D_{t} \leftarrow D_{t-1} \leftarrow D_{t-2}
    \rightarrow \hat{D}_{t-2} \rightarrow A^1_{t-2}$. Further, since
    $D_t$ is unobserved, no feasible state set
    exists except when $S_t$ includes the full history.

  \item[2. Retrospective price updates.] Suppose customers compare
    prices when making a decision: if the
    current price $A^2_{t}$ is lower than the previous price,
    customers may view this as a good bargain and book more
    (graphically we add the edge $A^2_{t-1} \textcolor{graphGreen}{
      \longrightarrow} B^2_{t+1}$). Shipping professionals may also
    factor this behavior into the price update ($A^2_{t-1}
    \textcolor{graphGreen}{\longrightarrow} A^2_{t}$).
  The resulting graph is denoted
  $\overset{\textcolor{graphGreen}{\longrightarrow}}{\gG_2}$. The
  degree of this confounding relationship is modeled by varying the
  boost to demand that occurs when $A^2_{t} < A^2_{t-1}$. Due to the
  founding path $B^2_{t+1} \textcolor{graphGreen}{\ldedge} A^2_{t-1}
  \textcolor{graphGreen}{\rdedge} A^2_t$, the
    state choice in \eqref{eq:dynamic_pricing_state} no longer
    satisfies Assumption \ref{assump:no-confounding-a}. Instead,
    including the previous price update in the state set (so $S_t =
    \{A^2_{t-1}, A^1_{t-1}, \hat{D}_{t-1}, B^1_t, \hat{D}_t\}$)
    satisfies all identifiability assumptions.

  \item[3. Competitor prices.] In this case, we return to the
    competitor price example from \Cref{subfig:competitor_price} where
    shipping professionals receive word-of-mouth intelligence on competitor
    prices. Clearly, competitor prices also impact the number of
    bookings received ($A^1_t
    \textcolor{graphRed}{\longleftrightarrow} B^1_{t+1}$). The graph
    is denoted
    $\overset{\textcolor{graphRed}{\longleftrightarrow}}{\gG_3}$. The
    degree of unmeasured confounding is modeled by varying the boost
    to demand that occurs when the shipping professional undercuts
    competitors and the probability of receiving accurate
    intelligence. In this case, no state choice satisfies Assumption
    \ref{assump:no-confounding-a}.
  \end{description}

\renewcommand{\arraystretch}{0.5}
\begin{table}[t]
  \centering
  \caption{Policy iteration results from $N=\num{100000}$
    episodes of length $T=500$.} 
  \label{table:dynamic_pricing_results}
  \begin{threeparttable}
    \begin{tabular}{>{\centering\arraybackslash}p{1.5cm}
      >{\centering\arraybackslash}p{1.5cm}
      >{\centering\arraybackslash}p{1.5cm}
      >{\centering\arraybackslash}p{2cm}
      >{\centering\arraybackslash}p{2cm}
      >{\centering\arraybackslash}p{2cm}
      >{\centering\arraybackslash}p{2cm}}
      \toprule \\
      & & \multicolumn{3}{c}{Mean cumulative reward (std. dev.)} & \multicolumn{2}{c}{Regret (\%)}  \\
      \cmidrule(lr){3-5} \cmidrule(lr){6-7}
      Graph & Degree & Null policy & PI with state
                              \eqref{eq:dynamic_pricing_state} & PI
                                                                 with
                                                                 correct
                     state & PI with state
                              \eqref{eq:dynamic_pricing_state} & PI
                                                                 with
                                                                 correct
      state \\ \midrule \\
      ${\gG}$                                                   &     -     &   1969.0   &  2056.0    & 2055.0 & 4.42   & 4.37    \\
      &           &   (87.37)         &  (75.56)          & (74.84)       &                                                           &                                                           \\ \addlinespace
$\overset{\textcolor{graphOrange}{\longrightarrow}}{\gG_1}$   &  0.1  &   1971.0   &  2058.0    & -                     & 4.41   & -                                                         \\
      &           &   (87.65)         &  (75.08)          &                       &                                                           &                                                           \\ \addlinespace
      &  0.5  &   1978.0   &  2063.0    & -                     & 4.3   & -                                                         \\
      &           &   (90.72)         &  (80.15)          &                       &                                                           &                                                           \\ \addlinespace
      &  0.9  &   1986.0   &  1987.0    & -                     & 0.05   & -                                                         \\
      &           &   (118.25)         &  (134.66)          &                       &                                                           &                                                           \\ \addlinespace
$\overset{\textcolor{graphGreen}{\longrightarrow}}{\gG_2}$    &  1  &   2227.0   &  2020.0    & 2451.0 & -9.3   & 10.06    \\
      &           &   (80.69)         &  (91.17)          & (86.81)       &                                                           &                                                           \\ \addlinespace
      &  2  &   2402.0   &  2041.0    & 3019.0 & -15.03   & 25.69    \\
      &           &   (82.25)         &  (93.41)          & (91.11)       &                                                           &                                                           \\ \addlinespace
      &  4  &   2415.0   &  2042.0    & 3911.0 & -15.45   & 61.95    \\
      &           &   (82.81)         &  (93.2)          & (87.05)       &                                                           &                                                           \\ \addlinespace
$\overset{\textcolor{graphRed}{\longleftrightarrow}}{\gG_3}$      &  1  &   2390.0   &  2621.0    &  -                    & 9.67   & -                                                         \\
      &           &   (84.99)         &  (82.39)          &                       &                                                           &                                                           \\ \addlinespace
      &  3  &   3002.0   &  2934.0    &  -                    & -2.27   & -                                                         \\
      &           &   (93.06)         &  (112.95)          &                       &                                                           &                                                           \\ \addlinespace
      &  5  &   3221.0   &  3112.0    &  -                    & -3.38   & -                                                         \\
      &           &   (96.22)         &  (120.95)          &                       &                                                           &                                                           \\ \addlinespace \addlinespace
      \bottomrule
  \end{tabular}
  \end{threeparttable}
\end{table}

\subsection{Results}

\Cref{table:dynamic_pricing_results} shows the mean cumulative reward
in $T= 500$ weeks for each above scenario. We use policy iteration to
train a policy using $N = 100,000$ episodes with the state set in
\eqref{eq:dynamic_pricing_state}. In the scenarios where a state set
satisfying all identifiability assumptions exists, we also run policy
iteration using that state set. We report the regret of these learned
policies in relation to the null policy that is used to generate the
training data.

In the basic scenario ($\gG$), the state set
\eqref{eq:dynamic_pricing_state} satisfies the identifiability
assumptions, and policy iteration leads to a modest improvement oer
the null policy. In the scenario with macroeconomic trend
($\overset{\textcolor{graphOrange}{\longrightarrow}}{\gG_1}$), policy
iteration with \eqref{eq:dynamic_pricing_state} also leads to a modest
improvement, but the improvement becomes smaller when the the
macroeconomic trend over time becomes stronger. In the scenario with
retrospective price updates
($\overset{\textcolor{graphGreen}{\longrightarrow}}{\gG_2}$), policy
iteration with the state set \eqref{eq:dynamic_pricing_state} actually
leads to inferior policy to the null. In contrast, if policy iteration
is run with the correct state set, a substantial improvement over the
null policy is observed. Finally, in the scenario with competitor
prices, $\overset{\textcolor{graphRed}{\longleftrightarrow}}{\gG_3}$,
whether or not policy iteration improves over the null policy depends
on the degree of confounding. These results highlight the importance
of choosing a good state set when using off-policy learning algorithms
in real-world problems.

\section{Discussion}
\label{sec:discussion}

In this paper we consider the identification problem arising in
off-policy learning for a general sequential decision problem. We present
graphical criteria for identifiability of the
value of an adaptive policy that require the state set to satisfy a
memorylessness property that is akin to the Markov property in the MDP
literature and an unconfoundedness property that extends Pearl's
back-door criterion for a static intervention. Our
results provide a unified causal framework for sequential decision
problems and generalize the common sequential ignorability assumption
in the DTR literature. Further, our framework reveals the
implicit assumption of randomized decisions in the current RL
literature and provides a principled basis for
wider application of RL to real-world problems.

A possible avenue of future work is to develop algorithms for state
variable selection in off-policy learning. This is closely related to confounder selection in causal inference
\parencite{guo_confounder_2023-1} and testing MDP assumptions
\parencite{shi2020does}. It is particularly interesting to
develop online RL algorithms (with some observational historical
data available) that can select state variables and improve the policy
at the same time. Another interesting problem is to extend the
graphical identifiability criteria when we can also access the natural
counterfactuals, a setting that is attracting some attention
recently \parencite{kallus2023efficient,stensrud2024optimal}. In the
DTR setting, this problem has been considered by
\textcite{young2014identification} and
\textcite{richardson_single_2013}.



\section*{Acknowledgements}
This work is in part supported by A.P. Moller Maersk, the
Innovation Fund Denmark, and the Engineering \& Physical Sciences
Research Council (EP/V049968/1).


\appendix

\section{Remarks on m-separation and no confounding}
For completeness we introduce some sets of walks. Let $W[J \mconn K
\mid L \ingraph{\gG}]$ denote the set of walks from $J$ to $K$ that is
not blocked given $L$.\footnote{More precisely, $W[J \mconn K \mid L \ingraph{\gG}]$ is a matrix with rows corresponding to vertices in
$J$, columns corresponding to vertices in $K$, and entries being a collection of walks from the corresponding row to the corresponding
  column. This allows us to represent walk concatenation as matrix multiplication. See \textcite{zhao_matrix_2024} for more detail.} Let
$W_s[J \mconn K \mid L \ingraph{\gG}]$ denote the subset of $W[J
\mconn K \mid L \ingraph{\gG}]$ that is simple. Similarly, let
$W_s[V_j \confpath V_k \mid L]$ denote the set of simple confounding
walks from $V_j$ to $V_k$ that is not blocked by $L$, and if this set
is empty we say $V_j$ and $V_k$ are \emph{unconfounded} given $L$ in
$\gG$. 
With this notation, confounder selection, with $V_j$ being the cause and $V_k$ being the effect under investigation, can be formulated as selecting $L \subset V$ such that
\begin{equation} \label{eq:no-confounding}
  V_j \nordpath L, V_k \nordpath L,~\text{and}~ W_s[V_j \confpath V_k
  \mid L] = \emptyset,
\end{equation}
that is, $L$ is not a descendant of $V_j$ or $V_k$ and all simple
confounding walks from $V_j$ to $V_k$ are blocked by $L$. \textcite{guo_confounder_2023} show that, under the first two conditions in \eqref{eq:no-confounding}, the third condition is
equivalent to the back-door criterion of \textcite{pearl_causal_1995},
which is a sufficient condition for the ignorability condition of
\textcite{rosenbaum83_centr_role_propen_score_obser} under a structural
interpretation of the graph as described in the main paper.

It is obvious that an ancestrally blocked (see main text) path is also blocked in our sense. Let $P[V_j \confpath V_k \mid_a L]$ be the set of confounding paths from $V_j$ to $V_k$ that is not ancestrally blocked by $L$. It is shown in \textcite[Lemma 2]{guo_confounder_2023} that
\[
  W_s[V_j \confpath V_k \mid L] = \emptyset\quad\text{if and only if}\quad P[V_j \confpath V_k \mid_a L] = \emptyset.
\]
This equivalence also holds for m-connection, directed paths, and confounding arcs.

\section{Technical proofs}\label{sec:technical-proofs}

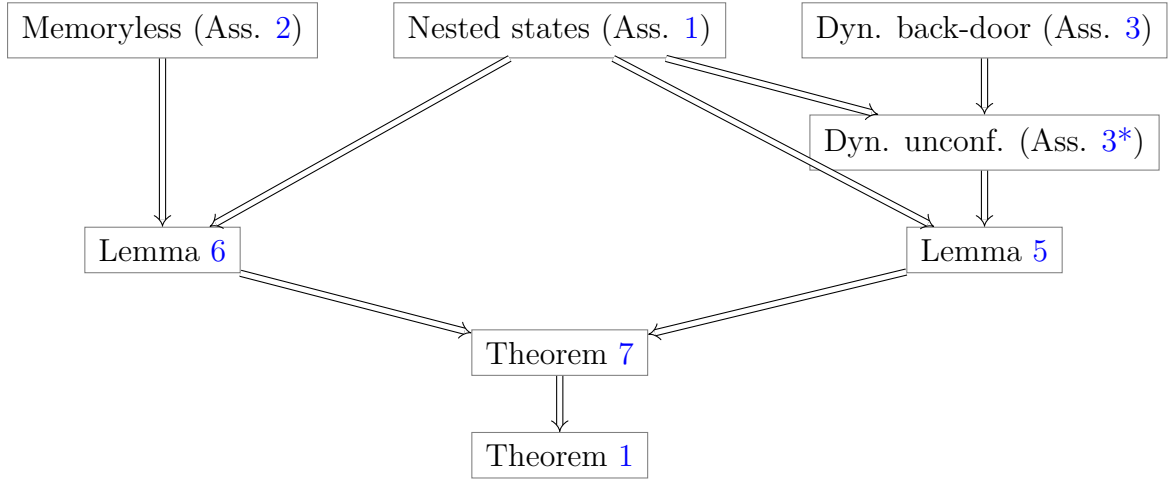
\begin{figure}[h!]
  \centering
  \vspace{1ex}
  \begin{tikzcd}[cells={nodes={draw=gray}}]
    \text{Memoryless (Ass. \ref{assump:memorylessness})} \arrow[dd, Rightarrow, black] & \text{Nested states (Ass. \ref{assump:nested-states})} \arrow[ldd, Rightarrow, black] \arrow[rd, Rightarrow, black]  \arrow[rdd, Rightarrow, black] & \text{Dyn. back-door (Ass. \ref{assump:no-confounding-a})} \arrow[d, Rightarrow, black] \\
    & & \text{Dyn. unconf. (Ass. \ref{assump:no-confounding-b})} \arrow[d, Rightarrow, black]  \\
    \text{\Cref{lemma:innovation_action_memorylessness_gt}} \arrow[rd, Rightarrow, black] & & \text{\Cref{lemma:innovation_action_unconfoundedness_gt}} \arrow[ld, Rightarrow, black] \\
    & \text{\Cref{thm:innovation_recursion}} \arrow[d, Rightarrow, black] & & \\
    & \text{\Cref{thm:main_result}} & &
  \end{tikzcd}
  \caption{Implications between assumptions, lemmas and theorems. Positivity and dynamic consistency are assumed to hold in all implications.}
  \label{fig:assump_implications}
\end{figure}

\begin{proof}[Proof of \Cref{prop:g-markov}]
  The construction of the new graph $\gG(g)$ describes the structure
  of the equations that define $V(g)$ and $V^-(g)$. The claim then
  follows from Theorem 3 in \textcite{zhao_admgs_2024} that shows that the
  nonparametric system of equations model implies the nested
  Markov model \parencite{richardson_nested_2023} which is global
  Markov by definition.
\end{proof}

\begin{proof}[Proof of \Cref{prop:dynamic-consistency}]
  Consider a topological ordering of $\gG(\underline{g}_{t-1})$ in the sense that $V_{j}(\underline{g}_{t-1}) \rdedge V_k(\underline{g}_{t-1}) \ingraph{\gG(\underline{g}_{t-1})}$ implies $V_j(\underline{g}_{t-1}) \prec V_k(\underline{g}_{t-1})$. Let $V_{j_1}(\underline{g}_{t-1}) \prec V_{j_2}(\underline{g}_{t-1}) \prec \dots \prec V_{j_p}(\underline{g}_{t-1})$ where every $V_{j_i}(\underline{g}_{t-1}) \in \{\underline{N}_t\cup \underline{A}_t\}(\underline{g}_{t-1})$. Clearly, since since $V_{j_1}(\underline{g}_{t-1})$ does not have any parents in $\{\underline{N}_t\cup \underline{A}_t\}(\underline{g}_{t-1})$, it takes the same value as $V_{j_1}(\underline{g}_{t})$ given that $A_{t-1}^{-}(\underline{g}_{t-1}) = A_{t-1}(\underline{g}_{t-1})$. This follows from simply replacing $A_{t-1}(\underline{g}_{t-1})$ by $A_{t-1}^{-}(\underline{g}_{t-1})$ in the structural equation and noticing that $A_{t-1}^{-}(\underline{g}_{t-1}) = A_{t-1}(\underline{g}_{t})$ (which is also equal to $A_{t-1}$). Considering $V_{j_2}(\underline{g}_{t-1})$, which may depend on $V_{j_1}(\underline{g}_{t-1})$, the arguments is the same. But in addition we use that $V_{j_1}(\underline{g}_{t-1}) = V_{j_1}(\underline{g}_{t})$ given that $A_{t-1}^{-}(\underline{g}_{t-1}) = A_{t-1}(\underline{g}_{t-1})$. The claim then follows from repeating this argument for $j_3, j_4, \dots, j_p$.
\end{proof}

\begin{lemma} \label{lem:vertex-map}
  For $1 \leq t \leq T$ and $1 \leq k \leq T - t +1$, consider the
  dynamic SWIG $\gG(\underline{g}_t)$ with vertices
  \[
    V(\gG(\underline{g}_{t})) = \overline{X}_{T+1}(\underline{g}_{t})
    \cup \overline{A}_T(\underline{g}_{t}) \cup \bigcup_{i=t}^T
    \{A_i^-(\underline{g}_{t})\},
  \]
  and another dynamic SWIG $\gG(\underline{g}_{t+k})$ with vertices
  \[
    V(\gG(\underline{g}_{t+k})) = \overline{X}_{T+1}(\underline{g}_{t+k})
    \cup \overline{A}_T(\underline{g}_{t+k}) \cup \bigcup_{i=t+k}^T
    \{A_i^-(\underline{g}_{t+k})\}.
  \]
  Given a walk $w$ in $\gG(\underline{g}_{t})$ that does not contain
  any $A_i(\underline{g}_t)$ for all $t \leq i \leq t+k-1$,
  the following map of vertices maps $w$ to a walk in
  $\gG(\underline{g}_{t+k})$:
  \begin{align*}
      &X_k(\underline{g}_t) \mapsto X_k(\underline{g}_{t+k}), ~ 1
        \leq k \leq T, \\
      &A_i^{-}(\underline{g}_t) \mapsto
      \begin{cases}
          A_i(\underline{g}_{t+k}), ~ t \leq i \leq t + k - 1, \\
          A_i^{-}(\underline{g}_{t+k}), ~ t+k \leq i \leq T,
      \end{cases} \\
      &A_i(\underline{g}_t) \mapsto A_i(\underline{g}_{t+k}), ~ 1
        \leq i \leq t-1 ~\text{and}~ t+k \leq i \leq T.
  \end{align*}
  In the image walk of $w$, every $A_i(\underline{g}_{t+k})$ for $t
  \leq i \leq t + k - 1$ is either an endpoint or
  a collider. Furthermore, if $w$ is a simple walk, its image is also a
  simple walk; if $w$ is a path, its image is also a path.
\end{lemma}
\begin{proof}
It is easy to verify using the definition of dynamic SWIGs that
edges in $\gG(\underline{g}_t)$ between the pre-images of the above vertex
map remain edges of the same type in $\gG(\underline{g}_{t+k})$
between the images of the vertex map. The claim that every $A_i(\underline{g}_{t+k})$ for $t \leq i \leq t + k - 1$ in the image of $w$ are either endpoints or colliders follows from the fact that $A_i^{-}(\underline{g}_t)$ only contains incoming edges in $\gG(\underline{g}_{t})$. The fact that a simple walk
(path) is mapped to a simple walk (path) follows from the
observation that the vertex map is injective.
\end{proof}

\begin{lemma} \label{lem:special-walk-1}
Let Assumption \ref{assump:nested-states} be given. Consider any
$Z_t(\underline{g}_t) \in A^{-}_t(\underline{g}_t) \cup (A_{t-1} \cup
S_{t-1}) \setminus S_t$. Suppose
\begin{equation}
  \label{eq:w-2}
   (\underline{N}_{t+1} \cup \underline{A}_{t+1})(\underline{g}_t)
   \mconn Z_t(\underline{g}_t) \mid (A_t \cup S_t)(\underline{g}_t).
\end{equation}
Then there exists a simple walk that does not contain any vertex in
$\underline{A}_{t+1}(\underline{g}_t)$ and looks like
\[
   \underline{N}_{t+1}(\underline{g}_t) \mconn
   Z_t(\underline{g}_t) \mid (A_t \cup S_t)(\underline{g}_t).
\]
\end{lemma}

\begin{proof}
  Consider a simple walk $w$ of the kind in \eqref{eq:w-2}. If $w$
  does not contain any vertex in
  $\underline{A}_{t+1}(\underline{g}_t)$, the claim follows
  immediately. Otherwise, let $A_{i}(\underline{g}_t)$, $t+1 \leq i
  \leq T$, to be the right-most such vertex in $w$. This shows that
  $w$ contains a
  subwalk $w'$ of the kind
  \[
      A_{i}(\underline{g}_t) \mconn Z_t(\underline{g}_t) \mid
      (A_t \cup S_t)(\underline{g}_t),
  \]
  and by construction $w'$ does not contain any vertex in
  $\underline{A}_{t+1}(\underline{g}_t)$ besides its left endpoint.
  Since $S_{i}(\underline{g}_t)$ contain the parent set of
  $A_{i}(\underline{g}_t)$, $w'$ must begin with
  $A_{i}(\underline{g}_t) \ldedge S_{i}(\underline{g}_t) \setminus
  S_t(\underline{g}_t)$ (we can substract $S_t(\underline{g}_t)$
  because otherwise the walk is blocked; $w'$ cannot start with
  $A_{i}(\underline{g}_t) \rdedge$ because it implies a
  directed walk from $A_{i}(\underline{g}_t)$ to $(Z_t \cup A_t \cup
  S_t)(\underline{g}_t)$ that goes back in time).
  By Assumption \ref{assump:nested-states}, we know
  \[
    S_{i}(\underline{g}_t) \setminus S_t(\underline{g}_t) \subseteq
  \underline{N}_{t+1}(\underline{g}_t) \cup
  \underline{A}_{t+1}(\underline{g}_t).
  \]
  Because $w'$ does not contain any vertex in
  $\underline{A}_{t+1}(\underline{g}_t)$ as a non-endpoint, this
  shows that $w'$ looks like
  \[
    A_i(\underline{g}_t) \ldedge
    \underline{N}_{t+1}(\underline{g}_t) \halfsquigfull \ast
      \fullsquigfull Z_t(\underline{g}_t) \mid (A_t \cup
      S_t)(\underline{g}_t).
    \]
  By considering the right most vertex in
  $\underline{N}_{t+1}(\underline{g}_t)$ that is contained in $w'$,
  we obtain a simple subwalk $w''$ of the kind
  \[
    \underline{N}_{t+1}(\underline{g}_t) \mconn Z_t(\underline{g}_t) \mid (A_t \cup
      S_t)(\underline{g}_t),
  \]
  so the claim follows.
\end{proof}

\begin{lemma}\label{lemma:innovation_action_unconfoundedness_gt_graph}
  Under Assumption \ref{assump:nested-states} and \ref{assump:no-confounding-b}, we have
    \begin{align*}
        \textnot (\underline{N}_{t+1} \cup \underline{A}_{t+1})(\underline{g}_t) \mconn A^{-}_t(\underline{g}_t) \mid (A_t \cup S_t)(\underline{g}_t).
    \end{align*}
\end{lemma}
\begin{proof}[Proof of
  \Cref{lemma:innovation_action_unconfoundedness_gt_graph}]
  If the claim is not true, by \Cref{lem:special-walk-1} there exists
  a simple walk $w$ of the kind
  \[
    \underline{N}_{t+1}(\underline{g}_t) \mconn
    A^{-}_t(\underline{g}_t) \mid (A_t \cup S_t)(\underline{g}_t).
  \]
  Because $\underline{N}_{t+1}(\underline{g}_t) \nordpath (A_t^- \cup
  A_t \cup S_t)(\underline{g}_t)$ and $A^{-}_t(\underline{g}_t)$ has
  no outgoing edges, $w$ must look like
  \[
    \underline{N}_{t+1}(\underline{g}_t) \confpath
    A^{-}_t(\underline{g}_t) \mid (A_t \cup S_t)(\underline{g}_t).
  \]
  Now consider the following cases:
    \begin{enumerate}
        \item $w$ contains $S_t(\underline{g}_t) \rdedge
          A_t(\underline{g}_t)$ or $A_t(\underline{g}_t) \ldedge
          S_t(\underline{g}_t)$.

        This contradicts the assumption that $w$ is not blocked by
        $S_t(\underline{g}_t)$.

      \item $w$ contains $\ldedge A_t(\underline{g}_t) \rdedge $.

        This contradicts the assumption that $w$ is not blocked by
        $A_t(\underline{g}_t)$.

      \item $w$ does not contain $A_t(\underline{g}_t)$.

        In this case, we may drop $A_t(\underline{g}_t)$ from the
        conditioning set, so $w$ looks like
        \[
          \underline{N}_{t+1}(\underline{g}_t) \confpath
        A^{-}_t(\underline{g}_t) \mid S_t(\underline{g}_t).
      \]
      By \Cref{lem:vertex-map}, $w$ can be mapped into a simple
      walk of the kind
        \[
            \underline{N}_{t+1}(\underline{g}_{t+1}) \confpath A_t \mid S_t \ingraph{\gG(\underline{g}_{t+1})}.
        \]
        However, this contradicts Assumption \ref{assump:no-confounding-b}.
    \end{enumerate}
    In conclusion, the existence of the walk $w$ contradicts the
    assumptions.
\end{proof}

\begin{lemma}\label{lemma:innovation_action_memorylessness_gt_graph}
  Under Assumptions
  \ref{assump:nested-states}, \ref{assump:memorylessness} and \ref{assump:no-confounding-b}
  we have
    \begin{align*}
        \textnot (\underline{N}_{t+1} \cup \underline{A}_{t+1})(\underline{g}_t) \mconn (A_{t-1} \cup S_{t-1}) \setminus S_t \mid (A_t \cup S_t)(\underline{g}_t)
    \end{align*}
\end{lemma}
\begin{proof}[Proof of
  \Cref{lemma:innovation_action_memorylessness_gt_graph}]

  If the claim is not true, by \Cref{lem:special-walk-1} there exists
  a simple walk of the kind
  \[
    \underline{N}_{t+1}(\underline{g}_t) \mconn
    (A_{t-1} \cup S_{t-1}) \setminus S_t \mid (A_t \cup S_t)(\underline{g}_t),
  \]
  that does not contain any vertices in
  $\underline{A}_{t+1}(\underline{g}_t)$.

  It is easy to see that any such simple walk does not contain $A_t(\underline{g}_t)$, because otherwise the walk is blocked (see the first two cases in the proof of \Cref{lemma:innovation_action_unconfoundedness_gt_graph}). Hence, we may drop $A_t(\underline{g}_t)$ from the conditioning set, and therefore there exists a simple walk of the kind
  \[
  \underline{N}_{t+1}(\underline{g}_t) \halfsquigfull \ast \fullsquighalf
  (A_{t-1} \cup S_{t-1}) \setminus S_t \mid S_t(\underline{g}_t).
  \]
  By \Cref{lem:vertex-map} (with $k = T - t + 1$), any simple walk of the above kind can be mapped into a simple walk that looks like
  \begin{align}\label{eq:all_innovations_simple_walk}
      \underline{N}_{t+1} \halfsquigfull \ast \fullsquighalf
          (A_{t-1} \cup S_{t-1}) \setminus S_t \mid S_t \ingraph{\gG}.
  \end{align}
  Note that by \Cref{lem:vertex-map}, every $A_i$ for $i \geq t$ in this image in \eqref{eq:all_innovations_simple_walk} can only be a collider. However, by Assumption \ref{assump:nested-states} we know $A_i \notin S_t$, so clearly no such simple walk can contain $A_i$ for any $i \geq t$.

  Next, let $j$ be the smallest integer in $t \leq j \leq T - 1$ such that there exists a simple walk in
  \eqref{eq:all_innovations_simple_walk} with $N_{j+1}$ as the left endpoint. That is, there exists a simple walk $w$ like
  \begin{align}\label{eq:innovation_simple_walk}
      N_{j+1} \halfsquigfull \ast \fullsquighalf
      (A_{t-1} \cup S_{t-1}) \setminus S_t \mid S_t \ingraph{\gG},
  \end{align}
  but
  \begin{align}\label{eq:no_innovation_simple_walk}
      \textnot N_l \halfsquigfull \ast \fullsquighalf (A_{t-1} \cup
    S_{t-1}) \setminus S_t \mid S_t \ingraph{\gG},~\text{for all}~t+1 \leq l \leq j.
  \end{align}
  By Assumption \ref{assump:memorylessness} we know $j \neq t$ so $j >
  t$.

  We will now show that this leads to a contradiction. We first make a simple observation. By
  Assumption \ref{assump:nested-states}, $S_{j} \subseteq S_{t} \cup A_{t} \cup X_{t+1}  \dots \cup A_{j-1} \cup X_{j}$, therefore
  \[
      (A_{t-1} \cup S_{t-1}) \setminus S_{t} \subseteq (A_{t-1} \cup
      S_{t-1}) \setminus S_j.
  \]
  So $w$ is also a simple walk in
  \[
      N_{j+1} \halfsquigfull \ast \fullsquighalf (A_{t-1} \cup S_{t-1}) \setminus S_{j} \mid S_t \ingraph{\gG}.
    \]
    Now consider the following cases:
  \begin{enumerate}
    \item  $w$ does not contain any non-endpoint in $(S_j \setminus S_t) \cup (S_t \setminus S_j)$.

    It is easy to see that $w$ is still unblocked given $S_j$ instead of $S_t$, so
    \begin{align*}\label{eq:innovation_simple_walk_Sj}
        N_{j+1} \halfsquigfull \ast \fullsquighalf (A_{t-1} \cup
      S_{t-1}) \setminus S_{j} \mid S_j \ingraph{\gG},
    \end{align*}
    which contradicts Assumption \ref{assump:memorylessness}.

    \item $w$ contains no non-endpoint in $S_j \setminus S_t$ but some non-endpoint in $S_t \setminus S_j$.

    Take $M$ to be the left-most such non-endpoint. Then $w$ looks like
    \[
        N_{j+1}  \halfsquigfull \ast \fullsquighalf M
        \halfsquigfull \ast \fullsquighalf (A_{t-1} \cup S_{t-1})
        \setminus S_{t} \mid S_{t}.
    \]
    Thus, there exists a simple subwalk $w'$ of $w$ that looks like
    \[
        N_{j+1}  \halfsquigfull \ast \fullsquighalf M \mid S_{t}.
    \]
    Since $w'$ does not contain any non-endpoint in $(S_j \setminus S_t) \cup (S_t \setminus S_j)$, $w'$ must belong to
    \[
        N_{j+1}  \halfsquigfull \ast \fullsquighalf M \mid S_{j}.
    \]
    But since $M \in S_t \setminus S_j$, this contradicts Assumption \ref{assump:memorylessness} again.

    \item $w$ contains some non-endpoint in $S_j \setminus S_t$.

    Let $L$ be such a non-endpoint. This means that $w$ looks like
    \[
        N_{j+1}  \mconn L \halfsquigfull \ast \fullsquighalf (A_{t-1} \cup S_{t-1}) \setminus S_{t} \mid S_{t}.
      \]
    Recall that $w$ does not contain any non-endpoint in $A_i$ for $i \geq t$. By Assumption \ref{assump:nested-states}, there must exist $t+1 \leq i \leq j$ such that $L \in N_i$, so
    \[
        N_{i} \halfsquigfull \ast \fullsquighalf (A_{t-1} \cup S_{t-1}) \setminus S_{t} \mid S_{t},
    \]
    which contradicts \eqref{eq:no_innovation_simple_walk}.
  \end{enumerate}
\end{proof}

\begin{proof}[Proof of \Cref{lemma:innovation_action_unconfoundedness_gt,lemma:innovation_action_memorylessness_gt}]
  \Cref{lemma:innovation_action_unconfoundedness_gt} follows from
  \Cref{lemma:innovation_action_unconfoundedness_gt_graph} and
  \Cref{prop:g-markov}. \Cref{lemma:innovation_action_memorylessness_gt}
  follows from \Cref{lemma:innovation_action_memorylessness_gt_graph}
  and \Cref{prop:g-markov}.
\end{proof}

\begin{proof}[Proof of \Cref{prop:strong-unconfoundedness-implies-weak}]
  Suppose there exists a walk $w$ that looks like
  \[
      N_{t+1}(\underline{g}_{k+1}) \confpath A_k \mid S_k
      \ingraph{\gG(\underline{g}_{k+1})},~\text{for some}~k \leq t \leq T.
  \]
  Consider the following cases.
  \begin{enumerate}
    \item $w$ contains $\rdedge A_i(\underline{g}_{k+1}) \ldedge$ for some $i \geq k+1$.

    Then $w$ looks like
    \[
        N_{t+1}(\underline{g}_{k+1}) \fullsquigfull \ast \fullsquighalf S_i(\underline{g}_{k+1}) \rdedge A_i(\underline{g}_{k+1}) \ldedge S_i(\underline{g}_{k+1}) \halfsquigfull \ast \fullsquigfull A_k \mid S_k \ingraph{\gG(\underline{g}_{k+1})}, k\leq t \leq T.
    \]
    But since $A_i(\underline{g}_{k+1}) \notin S_k$ this walk is blocked.

    \item $w$ contains $S_i(\underline{g}_{k+1}) \rdedge
      A_i(\underline{g}_{k+1}) \rdedge$ for some $i \geq k+1$.

    Then $w$ looks like
    \[
        N_{t+1}(\underline{g}_{k+1}) \fullsquigfull \ast \fullsquighalf S_i(\underline{g}_{k+1}) \rdedge A_i(\underline{g}_{k+1}) \nosquigfull \ast \fullsquigfull A_k \mid S_k \ingraph{\gG(\underline{g}_{k+1})}.
    \]
    However, this contradicts $A_i(\underline{g}_{k+1}) \nordpath A_k \cup S_k$.

    \item $w$ contains $\ldedge A_i(\underline{g}_{k+1}) \rdedge$ for
      some $i \geq k+1$.

    Then $w$ looks like
    \[
        N_{t+1}(\underline{g}_{k+1}) \fullsquigfull \ast \fullsquigno A_i(\underline{g}_{k+1}) \rdpath \ast \fullsquigfull  A_k \mid S_k \ingraph{\gG(\underline{g}_{k+1})},
    \]
    Same as 2, this contradicts $A_i(\underline{g}_{k+1}) \nordpath A_k \cup S_k$.

    \item $w$ contains $\ldedge A_i(\underline{g}_{k+1}) \ldedge
      S_i(\underline{g}_{k+1})$ for some $i \geq k+1$.

    By considering the right-most $A_i(\underline{g}_{k+1})$ in this
    walk, $w$ has a sub-walk that looks like
    \[
        A_i(\underline{g}_{k+1}) \ldedge S_i(\underline{g}_{k+1})
        \setminus S_k \halfsquigfull \ast \fullsquigfull A_k \mid S_k \ingraph{\gG(\underline{g}_{k+1})},
      \]
    and we can subtract $S_k$ here because the walk is unblocked given $S_k$.
    Since by Assumption \ref{assump:nested-states},
    \[
        S_i(\underline{g}_{k+1}) \subseteq (S_k \cup A_k \cup X_{k+1} \cup \dots, A_{i-1} \cup X_i)(\underline{g}_{k+1}),
    \]
    and $S_i(\underline{g}_{k+1}) \setminus S_k \nordpath A_k \cup
    S_k$, this subwalk must look like
    \[
        A_i(\underline{g}_{k+1}) \ldedge S_i(\underline{g}_{k+1})
        \setminus S_k
        \confpath A_k \mid S_k \ingraph{\gG(\underline{g}_{k+1})}.
      \]
    Note that if for some variable $L \in X_l$ we have $L \in
    S_i(\underline{g}_{k+1})$, then by Assumption
    \ref{assump:nested-states} we have $L \in
    S_l(\underline{g}_{k+1})$. Thus, there exists a simple walk $w'$
    that does not contain $A_i(\underline{g}_{k+1})$ and is like
    \[
        N_l(\underline{g}_{k+1}) \fullsquigfull \ast \fullsquigfull A_k \mid S_k \ingraph{\gG(\underline{g}_{k+1})},
    \]
    for some $k + 1 \leq l \leq i$. By \Cref{lem:vertex-map} $w'$ maps to
    a simple walk in $\gG$ of the kind
    \[
        N_l \fullsquigfull \ast \fullsquigfull A_k \mid S_k \ingraph{\gG},
    \]
    which contradicts Assumption \ref{assump:no-confounding-a}.

    \item $w$ contains no $A_i(\underline{g}_{k+1})$ for $i \geq k+1$.

    Then by \Cref{lem:vertex-map}, $w$ maps to a walk in $\gG$ that looks like
    \[
        N_{t+1} \confpath A_k \mid S_k \ingraph{\gG}, ~ k \leq t \leq T.
    \]
    However, this immediately contradicts Assumption \ref{assump:no-confounding-a}.
  \end{enumerate}
  In each case above, a contradiction is obtained. Thus, Assumption
  \ref{assump:no-confounding-b} must hold.
\end{proof}

\begin{proof}[Proof of \Cref{prop:seq_ignorability_and_dyn_unconfoundedness}]
  We first show that Assumption \ref{assump:no-confounding-b} implies
  sequential ignorability, that is, the ``$\Rightarrow$''
  direction in this Proposition. Suppose there is a walk $w$ that looks like
  \[
      R_{T+1}(\underline{g}_1) \mconn A_t^{-}(\underline{g}_1) \mid
      S_t(\underline{g}_1) \cup
      \overline{A}^{-}_{t-1}(\underline{g}_1) \ingraph{\gG(\underline{g}_1)}.
  \]
  Since $R_{T+1}(\underline{g}_1) \nordpath S_t(\underline{g}_1) \cup \overline{A}^{-}_{t}(\underline{g}_1)$ and $A_t^{-}(\underline{g}_1)$ has no outgoing arrows, $w$ must be of the kind
  \[
      R_{T+1}(\underline{g}_1) \confpath A_t^{-}(\underline{g}_1) \mid
      S_t(\underline{g}_1) \cup
      \overline{A}^{-}_{t-1}(\underline{g}_1) \ingraph{\gG(\underline{g}_1)}.
  \]
  Since $S_i(\underline{g}_1)$
  contain the parent set of $A_i(\underline{g}_1)$, $w$ cannot contain
  any $A_i(\underline{g}_1)$ for $i < t$; otherwise it would contain
  either $S_i(\underline{g}_1) \rdedge A_i(\underline{g}_1)$ or
  $\ldedge A_i(\underline{g}_1) \rdedge$ and both would be
  blocked (because $(S_i \cup A_i)(\underline{g}_1) \subseteq
  S_t(\underline{g}_1)$ by Assumption \ref{assump:DTR}). Thus, by
  \Cref{lem:vertex-map}, $w$ can be mapped to a walk $w'$
  in $\gG(\underline{g}_t)$ that looks like
  \[
      R_{T+1}(\underline{g}_t) \confpath A_t^{-}(\underline{g}_t) \mid S_t
      \ingraph{\gG(\underline{g}_t)}.
  \]
  Next, suppose that $w'$ contains $A_t(\underline{g}_t)$. Then since $A_t(\underline{g}_t) \notin S_t$, it cannot be a collider in $w'$. Furthermore, since $S_t$ contain the parent set of $A_t(\underline{g}_t)$, $w'$ must look like
  \begin{equation}
    \label{eq:w-prime}
      R_{T+1}(\underline{g}_t) \fullsquigfull \ast \fullsquigno
      A_t(\underline{g}_t) \nosquigfull \ast \fullsquigfull
      A_t^{-}(\underline{g}_t)
      \mid S_t \ingraph{\gG(\underline{g}_t)}.
  \end{equation}
  However, since
  \[
      A_t(\underline{g}_t) \nordpath (A_t^{-}\cup S_t)(\underline{g}_t),
  \]
  \eqref{eq:w-prime} is not possible. Hence, $w'$ cannot contains $A_t(\underline{g}_t)$, and by \Cref{lem:vertex-map}, $w$ can be mapped to a walk in $\gG(\underline{g}_{t+1})$ of the kind
  \begin{equation}\label{eq:seq_ignorability_type_walk}
      R_{T+1}(\underline{g}_{t+1}) \confpath A_t \mid S_t \ingraph{\gG(\underline{g}_{t+1})}, ~ \text{for} ~ t\in [T].
  \end{equation}
  However, this contradicts Assumption \ref{assump:no-confounding-b}.

  We now show the reverse implication ``$\Leftarrow$'' given
  \eqref{eq:directed-path-state-reward}, which says
  $V_j(g) \rdpath R_{T+1}(g)$ for all $V_j \in
  \overline{S}_T$. Suppose there exists a walk $w$ that is like
  \[
      N_{t+1}(\underline{g}_{k+1}) \confpath A_k \mid S_k \ingraph{\gG(\underline{g}_{t+1})}, ~ k\leq t\leq T.
  \]
  By \eqref{eq:directed-path-state-reward}, this implies that
  \[
      R_{T+1}(\underline{g}_{k+1}) \ldpath N_{t+1}(\underline{g}_{k+1}) \confpath A_k \mid S_k \ingraph{\gG(\underline{g}_{k+1})}, ~ k\leq t \leq T,
  \]
  which clearly also contradicts Assumption \ref{assump:no-confounding-b}.
\end{proof}

\begin{proof}[Proof of \Cref{prop:dynamic_back_door_no_At_in_St}]
  It suffices to prove that \eqref{eq:dynamic_back_door_t_only}
  implies Assumption \ref{assump:no-confounding-a}.
  Consider the induction hypothesis indexed by $0 \leq m \leq T-1$:
  \[
      \text{there is no simple walk like}~N_{t+1} \confpath A_k \mid S_k \ingraph{\gG}~\text{for
        all}~1 \leq k \leq t \leq T~\text{such that}~t \leq k + m.
    \]
  Equation \eqref{eq:dynamic_back_door_t_only} is equivalent to this
  hypothesis with $m = 0$ and Assumption \ref{assump:no-confounding-a}
  is equivalent to this hypothesis with $m = T-1$. We next prove that if
  this hypothesis is true for some $0 \leq m \leq T-2$, it is also true for
  $m + 1$.

  Suppose this hypothesis is not true for $m+1$, then there exists $t
  = k + m + 1$ and a simple walk $w$ like
  \[
      N_{t+1} \confpath A_k \mid S_k \ingraph{\gG}.
  \]
  Consider the following cases.
  \begin{enumerate}
    \item $w$ does not contain any vertex in $(S_t \setminus S_k) \cup (S_k \setminus S_t)$.

    Then we are free to condition on $S_t$ instead of $S_k$ and $w$ looks like
    \[
        N_{t+1} \confpath A_k \mid S_{t}
    \]
    But this immediately contradicts memorylessness (Assumption
    \ref{assump:memorylessness}).

  \item $w$ contains some vertex $L \in S_t \setminus S_k$.

    By Assumption \ref{assump:nested-states} and the assumption that
    the state contains no previous decisions, we know
    \[L \in S_t \setminus S_k
      \subseteq N_{k+1} \cup \dots \cup N_t
    \]
    so $L \nordpath A_k \cup S_k$. 
    and $w$ must look like
    \[
        N_{t+1} \fullsquigfull \ast \fullsquighalf L \confpath A_k
        \mid S_k.
    \]
    By choosing the right-most such $L$, we obtain a simple walk in
    \[
        N_{k+1} \cup \dots \cup N_t \confpath A_k \mid S_k,
    \]
    which contradicts the induction hypothesis.

    \item $w$ does not contain any vertex in $S_t \setminus S_k$, but
      does contain some vertex $M \in S_k \setminus S_t$.

    Since $M \in S_k$, $M$ must be a collider (otherwise $w$ is
    blocked by $S_k$). So $w$ looks like
    \[
        N_{t+1} \confpath M \confpath A_k \mid S_k.
    \]
    By taking $M$ to be the left-most such vertex, this implies the
    existence of a simple subwalk $w'$ that looks like
    \[
        N_{t+1} \confpath S_k \setminus S_t \mid S_{k},
    \]
    that does not contain any vertex in $(S_t \setminus S_k) \cup (S_k
    \setminus S_t)$. Thus, we are free to condition on $S_t$, and $w'$
    looks like
    \[
      N_{t+1} \confpath S_k \setminus S_t \mid S_{t}.
    \]
    This contradicts memorylessness (Assumption \ref{assump:memorylessness}).
  \end{enumerate}
\end{proof}





\section{Simulation Details}
Here we outline the details of the dynamic pricing simulation from \Cref{sec:simulation_study}.

\subsection{The graph simulator package}
For the purpose of simulating data from DAGs, we created a simple Python module \texttt{graph\_simulator}\footnote{\url{https://pypi.org/project/graph-simulator/}} implemented in C++. Dependencies between vertices in the graph are specified in \texttt{YAML} file. For example, the simple graph $X \rightarrow Y$ is specified by the YAML code below:
\begin{tcolorbox}[colframe=white, colback=white, boxrule=0mm, sharp corners=south]
    \begin{lstlisting}[basicstyle=\bfseries\ttfamily, frame=none, breaklines=true, keywordstyle=\color{blue}]
    X:
        kernel:
          type: "uniform"
          sample_domain: [0, 1]
          terms: null
        dependencies: null
    Y:
        kernel:
          type: "linear"
          sample_domain: [1, 1.5]
          noise: 0.1
          terms:
            - intercept: 1
              indicators: null
              value: 0.5
              variable:
                1: "X"
        dependencies:
        1: ["X"]
    \end{lstlisting}
\end{tcolorbox}
Here $X$ does not have any parents as specified by ``\texttt{dependencies:null}'', and it is uniformly distributed with support $\{0, 1\}$. This means that at every time-point, we sample $X_t$ from $\textup{Unif}(\{0, 1\})$. In contrast $Y_t$ has a parent $X_{t-1}$ (graphically we write $X_{t-1} \rdedge Y_t$) which is specified by \texttt{1:["X"]} -- interpreted as ``X of lag 1''-- under \texttt{dependencies}. In the \texttt{kernel} section, we see that $Y$ depends on its parents linearly (\texttt{type:"linear"}), but takes a random value with probability 0.1 in $\{1, 1.5\}$, as indicated by \texttt{noise:0.1} and \texttt{sample domain:[1, 1.5]}. The linear dependency is specified in the \texttt{terms} subsection. In this case, there is a single term composed of an intercept of value 1 (\texttt{intercept:1}), plus a value of 0.5 (\texttt{value:0.5}) times the values of ``$X$ at lag 1'' (\texttt{1:"X"} under \texttt{variable}). Terms may in addition depend of indicator functions that can render the whole term zero. Thus, in the above example we may write the functional form of $Y_t$ as
\[
    Y_t =
    \begin{cases}
        1 + 0.5 \cdot X_{t-1} ~ \textup{with probability} ~ 0.9 \\
        \textup{Unif}(\{1, 1.5\}) ~ \textup{with probability} ~ 0.1
    \end{cases}.
\]

\subsection{Dynamic pricing setup}
At each time-step $t$ we generate variables in the following order: the latent demand $D_t$, the company estimate $\hat{D}_t$, the realized bookings $B^1_t$ and $B^2_t$, the reward $R_t$, the unobserved competitor price $A^{c, 1}_t$ and $A^{c, 2}_t$, and the initial and updated company prices $A^1_t$, $A^2_t$.

\subsubsection*{Demand}
The latent demand is either uniformly distributed or equal to its previous value depending on a Bernoulli draw
\[
    D_t =
    \begin{cases}
        \sim \textup{Unif}(\{0, 1\})    & \textup{w/p \quad $p_D$} \\
        D_{t-1}                             & \textup{w/p \quad $1 - p_D$} \\
    \end{cases}.
\]
We assume that the company's ability to infer $D_t$ occurs with a fixed accuracy given by $p_{\hat{D}}$, that is $\hat{D}_t$ is equal to $D_t$ with probability $p_{\hat{D}}$ and otherwise it is uniformly distributed over $\{0, 1\}$. In the basic graph $\gG$ we have $p_{\hat{D}} = 1$, but in the ``macroeconomics trend scenario'' (graph $\overset{\textcolor{graphOrange}{\longrightarrow}}{\gG_1}$) we assume $p_{\hat{D}} < 1$.

\subsubsection*{Bookings}
The total capacity on each vessel is assumed to be fixed ($C=6$) so that
\[
    B_{t}^1 + B_{t+1}^2 \leq C.
\]
For each vessel departure the unconstrained bookings $\tilde{B}_{t}^1$ and $\tilde{B}^2_{t+1}$ (i.e. the bookings that would have been realized with no capacity constraint) are assumed to be given by a mixed distribution,
\[
    \tilde{B}_{t}^i  =
    \begin{cases}
        \sim \textup{Poisson}(\lambda^i_{t}) ~ \textup{w/p $1 - p_{B^i}$} \\
        \sim \textup{Unif}(\{0, 1, \dots, C\}) ~ \textup{w/p $1 - p_{B^i}$}
    \end{cases}
\]
where the intensity parameter $\lambda^i_{t}$ is specific to the departure and time
\[
    \lambda^i_{t} = \textup{exp}\left\{\beta^i_0 + (\beta^i_1 + \beta^i_2 D_{t-i}) (A^i_{t-1})^2 + \beta^i_3 \mathds{1}(A^i_{t-1} < A^{c, i}_{t-1}) + \beta^i_4 \mathds{1}(A^i_{t-1} > A^i_{t-2})\right\}.
\]
In the basic scenario $\gG$, we assume that $\beta^i_3 = \beta^i_4 = 0$, so that bookings only depend on the latent demand for the departure $D_{t-i}$ and the price the customer faces $A^i_{t-1}$. In the ``competitor price scenario'' (graph $\overset{\textcolor{graphRed}{\longrightarrow}}{\gG_3}$), we in addition assume that $\beta^1_{3} > 0$, so that in the case the company undercuts the competitors. This boosts the demand by $\beta^1_{3}$. In the ``retrospective price update'' (graph $\overset{\textcolor{graphGreen}{\longrightarrow}}{\gG_2}$) scenario, we assume that $\beta^1_4 > 0$, so that customers react to prices being lower than usual at the time of the price update. The realized bookings are simply truncated in the following way
\[
    B_{t}^1 = \textup{min}\{C, \tilde{B}_{t}^1\}, \quad \textup{and} \quad B_{t+1}^2 = \textup{min}\{C - B_{t}^{1}, \tilde{B}_{t+1}^2\}.
\]

\subsubsection*{Reward}
The intermediate reward is given by
\[
  R_{t+1} = A^2_t B^2_{t+1} + A^1_t B^1_{t+1}.
\]

\subsubsection*{Prices}
For simplicity we assume that the company and competitors may set three prices - ``low'', ``medium'' and ``high''. We treat $A^{c, 1}_t$ and $A^{c, 2}_t$ as an exogenous variables and let them be uniformly distributed.
\[
    A^{c, i}_t \sim \textup{Unif}(\{1, 2, 3\}).
\]
The company prices $A^1_t$, $A^2_t$ are assumed to depend on the previous prices, the observed bookings and the demand estimate through the mixed form
\[
  A^i_t =
    \begin{cases}
        \alpha^i_0 + \alpha^i_1 \hat{D}_t + (\alpha^i_2 + \alpha^i_3 A^1_{t-1}) \mathds{1}\{(B^1_t > \bar{B}^1) \cap (A^1_{t-1} > \bar{P}^1)\} + \alpha^i_4 A^2_{t-1}  & \textup{w/p $p_{A^{c, i}}(1 - p_{A^i})$} \\
        \xi^i_0 + \xi^i_1 A^{c, 1}_t & \textup{w/p $(1 - p_{A^{c, i}})(1 - p_{A^i})$} \\
        \sim \textup{Unif}(\{1, 2, 3\}) &\textup{w/p} ~ p_{A^i}.
    \end{cases}
\]
In the basic graph $\gG$ we have $p_{A^{c, i}} = 1$ and $\alpha_4^i = 0$, so that prices may only rely the demand estimate $\hat{D}_t$ and the previous initial price $A^1_{t-1}$ and the bookings realized at that price $B^1_t$.
For the initial price we assume $\alpha^1_0 = \alpha^1_1 = \alpha^1_2 = 1$ and $\alpha^1_3 = \alpha^1_4 = 0$, that is, if the demand estimate $\hat{D}_t$ is higher, this increases the initial price, and similarly, if a sufficient combination of price and bookings is met ($\mathds{1}\{(B^1_t > \bar{B}^1) \cap (A^1_{t-1} > \bar{P}^1)\}$) for the previous vessel departure, the initial price is set higher. In the ``competitor price scenario'' (graph $\overset{\textcolor{graphRed}{\longrightarrow}}{\gG_3}$) we assume that with some probability $(1 - p_{A^{c, i}})(1 - p_{A^i})$ (i.e. $p_{A^{c, i}} < 1$), the company has access to an accurate estimate of the initial competitor price $A_t^{c, 1}$, in which case, the initial price is set to undercut the competitor price ($\xi^i_0=-1$ and $\xi^i_1 = 1$).
For the price update we assume $\alpha^2_0 = \alpha^2_3 = 1$ and $\alpha^2_1 = \alpha^2_2 = \alpha^2_4 = 0$, so that the price update follows a simple rule of increasing the initial price $A^1_{t-1}$ by 1, if the sufficient combination of price and bookings ($\mathds{1}\{(B^1_t > \bar{B}^1) \cap (A^1_{t-1} > \bar{P}^1)\}$) is met. In the ``retrospective price update scenario'' (graph $\overset{\textcolor{graphGreen}{\longrightarrow}}{\gG_1}$)  we in addition let $\alpha^2_4 = -1$, so that shipping professionals generally try to undercut the previous price update level, so attract more demand.

To see the full parametrization, see \Cref{table:sim_param_values}.
\renewcommand{\arraystretch}{1.1}
\begin{table}
  \centering
  \caption{Simulation parameter values}
  \label{table:sim_param_values}
  \begin{threeparttable}
    \footnotesize
    \begin{tabular}{c c c c c c c c c c c c c c} \toprule
      & \multicolumn{13}{c}{Graph} \\
      & $\gG$ & \multicolumn{4}{c}{$\overset{\textcolor{yellow}{\longrightarrow}}{\gG}$} & \multicolumn{4}{c}{$\overset{\textcolor{green}{\longrightarrow}}{\gG}$} & \multicolumn{4}{c}{$\overset{\textcolor{red}{\longrightarrow}}{\gG}$} \\
      \cmidrule(lr){3-6} \cmidrule(lr){7-10} \cmidrule(lr){11-14}
      Degree & - & 0 & 0.1 & 0.5 & 0.9 & 0 & 1 & 2 & 4 & 0 & 1 & 3 & 5\\ \midrule
      $C$ & 6 & 6 & 6 & 6 & 6 & 6 & 6 & 6 & 6 & 6 & 6 & 6 & 6 \\
      $p_D$ & 1 & 1.0 & 0.9 & 0.5 & 0.1 & 1 & 1 & 1 & 1 & 1 & 1 & 1 & 1 \\
      $p_{\hat{D}}$ & 0.25 & 0.25 & 0.25 & 0.25 & 0.25 & 0.25 & 0.25 & 0.25 & 0.25 & 0.25 & 0.25 & 0.25 & 0.25 \\
      $p_{B^{1}}$ & 0.15 & 0.15 & 0.15 & 0.15 & 0.15 & 0.15 & 0.15 & 0.15 & 0.15 & 0.15 & 0.15 & 0.15 & 0.15 \\
      $\beta^1_0$ & 1 & 1 & 1 & 1 & 1 & 1 & 1 & 1 & 1 & 1 & 1 & 1 & 1 \\
      $\beta^1_1$ & -0.65 & -0.65 & -0.65 & -0.65 & -0.65 & -0.65 & -0.65 & -0.65 & -0.65 & -0.65 & -0.65 & -0.65 & -0.65 \\
      $\beta^1_2$ & 0.5 & 0.5 & 0.5 & 0.5 & 0.5 & 0.5 & 0.5 & 0.5 & 0.5 & 0.5 & 0.5 & 0.5 & 0.5 \\
      $\beta^1_3$ & 0 & 0 & 0 & 0 & 0 & 0 & 0 & 0 & 0 & 0 & 1 & 3 & 5 \\
      $\beta^1_4$ & 0 & 0 & 0 & 0 & 0 & 0 & 0 & 0 & 0 & 0 & 0 & 0 & 0 \\
      $p_{B^{2}}$ & 0.15 & 0.15 & 0.15 & 0.15 & 0.15 & 0.15 & 0.15 & 0.15 & 0.15 & 0.15 & 0.15 & 0.15 & 0.15 \\
      $\beta^2_0$ & 1 & 1 & 1 & 1 & 1 & 1 & 1 & 1 & 1 & 1 & 1 & 1 & 1 \\
      $\beta^2_1$ & -0.65 & -0.65 & -0.65 & -0.65 & -0.65 & -0.65 & -0.65 & -0.65 & -0.65 & -0.65 & -0.65 & -0.65 & -0.65 \\
      $\beta^2_2$ & 0.2 & 0.2 & 0.2 & 0.2 & 0.2 & 0.2 & 0.2 & 0.2 & 0.2 & 0.2 & 0.2 & 0.2 & 0.2 \\
      $\beta^2_3$ & 0 & 0 & 0 & 0 & 0 & 0 & 0 & 0 & 0 & 0 & 0 & 0 & 0 \\
      $\beta^2_4$ & 0 & 0 & 0 & 0 & 0 & 0 & 1 & 2 & 4 & 0 & 0 & 0 & 0 \\
      $p_{A^{1}}$ & 0.15 & 0.15 & 0.15 & 0.15 & 0.15 & 0.15 & 0.15 & 0.15 & 0.15 & 0.15 & 0.15 & 0.15 & 0.15 \\
      $p_{A^{c, 1}}$ & 1 & 1 & 1 & 1 & 1 & 1 & 1 & 1 & 1 & 1 & 0.5 & 0.25 & 0.1 \\
      $\alpha^1_0$ & 1 & 1 & 1 & 1 & 1 & 1 & 1 & 1 & 1 & 1 & 1 & 1 & 1 \\
      $\alpha^1_1$ & 1 & 1 & 1 & 1 & 1 & 1 & 1 & 1 & 1 & 1 & 1 & 1 & 1 \\
      $\alpha^1_2$ & 1 & 1 & 1 & 1 & 1 & 1 & 1 & 1 & 1 & 1 & 1 & 1 & 1 \\
      $\alpha^1_3$ & 0 & 0 & 0 & 0 & 0 & 0 & 0 & 0 & 0 & 0 & 0 & 0 & 0 \\
      $\alpha^1_4$ & 0 & 0 & 0 & 0 & 0 & 0 & 0 & 0 & 0 & 0 & 0 & 0 & 0 \\
      $\xi^1_0$ & -1 & -1 & -1 & -1 & -1 & -1 & -1 & -1 & -1 & -1 & -1 & -1 & -1 \\
      $\xi^1_1$ & 1 & 1 & 1 & 1 & 1 & 1 & 1 & 1 & 1 & 1 & 1 & 1 & 1 \\
      $p_{A^{2}}$ & 0.15 & 0.15 & 0.15 & 0.15 & 0.15 & 0.15 & 0.15 & 0.15 & 0.15 & 0.15 & 0.15 & 0.15 & 0.15 \\
      $p_{A^{c, 2}}$ & 1 & 1 & 1 & 1 & 1 & 1 & 1 & 1 & 1 & 1 & 1 & 1 & 1 \\
      $\alpha^2_0$ & 1 & 1 & 1 & 1 & 1 & 1 & 1 & 1 & 1 & 1 & 1 & 1 & 1 \\
      $\alpha^2_1$ & 0 & 0 & 0 & 0 & 0 & 0 & 0 & 0 & 0 & 0 & 0 & 0 & 0 \\
      $\alpha^2_2$ & 0 & 0 & 0 & 0 & 0 & 0 & 0 & 0 & 0 & 0 & 0 & 0 & 0 \\
      $\alpha^2_3$ & 1 & 1 & 1 & 1 & 1 & 1 & 1 & 1 & 1 & 1 & 1 & 1 & 1 \\
      $\alpha^2_4$ & 0 & 0 & 0 & 0 & 0 & 0 & -1 & -1 & -1 & 0 & 0 & 0 & 0 \\
      $\xi^2_0$ & 0 & 0 & 0 & 0 & 0 & 0 & 0 & 0 & 0 & 0 & 0 & 0 & 0 \\
      $\xi^2_1$ & 0 & 0 & 0 & 0 & 0 & 0 & 0 & 0 & 0 & 0 & 0 & 0 & 0 \\
      \bottomrule
      \end{tabular}
  \end{threeparttable}
\end{table}



\printbibliography

\end{document}